%% file: ms.tex
\documentclass{sig-alternate-10pt}
\usepackage{amsfonts}
\usepackage{mathtools}
\usepackage{cite}
\usepackage{booktabs}
\usepackage{hyperref}
\usepackage{soul}
\usepackage{color}
\usepackage{thm-restate}
\usepackage{comment}
\usepackage[linesnumbered,ruled]{algorithm2e}
\usepackage{caption}
\usepackage{amsmath}

\newtheorem{theorem}{Theorem}

\newtheorem{lemma}[theorem]{Lemma}

\newenvironment{pfsketch}{\noindent{Proof Sketch}\hspace*{1em}}{\qed\bigskip}
\newcommand{\hi}[1]{\colorbox{yellow}{#1}}

\begin{document}

\title{OCCAM: An Optimization-Based Approach to Network Inference}

\numberofauthors{3} %  in this sample file, there are a *total*
% of EIGHT authors. SIX appear on the 'first-page' (for formatting
% reasons) and the remaining two appear in the \additionalauthors section.
%
\author{
  % You can go ahead and credit any number of authors here,
  % e.g. one 'row of three' or two rows (consisting of one row of three
  % and a second row of one, two or three).
  %
  % The command \alignauthor (no curly braces needed) should
  % precede each author name, affiliation/snail-mail address and
  % e-mail address. Additionally, tag each line of
  % affiliation/address with \affaddr, and tag the
  % e-mail address with \email.
  %
  % 1st. author
  \alignauthor
Anirudh Sabnis\\
  \affaddr{UMass Amherst}\\
  \email{asabnis@cs.umass.edu}
  \alignauthor
 Ramesh K. Sitaraman\\
  \affaddr{UMass Amherst}\\
  \email{ramesh@cs.umass.edu}
  \alignauthor
 Donald Towsley\\
  \affaddr{UMass Amherst}\\
  \email{towsley@cs.umass.edu}
}
\maketitle
\input{abstract}
\input{introduction}
\input{theory}
\input{empirical}
\input{tree}

\input{conclusion}

\bibliographystyle{abbrv}
\bibliography{ms}
\appendix
\input{appendix1}

\input{appendix2}

\end{document}

%% file: abstract.tex
\begin{abstract}
We study the problem of inferring the structure of a communication network based only on network measurements made from a set of hosts situated at the network periphery.  Our novel approach called ``OCCAM'' is based on the principle of occam's razor and finds the ``simplest'' network that explains the observed network measurements.  OCCAM infers the internal topology of a communication network, including the internal nodes and links of the network that are not amenable to direct measurement. In addition to network topology, OCCAM infers the routing paths that packets take between the hosts. OCCAM uses  path metrics measurable from the hosts and expresses the observed measurements as constraints of a mixed-integer bilinear optimization problem that can then be feasibly solved to yield the network topology and the routing paths. We empirically validate OCCAM on a wide variety of real-world ISP networks and show that its inferences agree closely with the ground truth. Specifically, OCCAM infers the topology with an average network similarity score of 93\% and infers routing paths with a path edit distance of 0.20. Further, OCCAM is robust to error in its measured path metric inputs, producing high quality inferences even when 20-30\% of  its inputs are erroneous. Our work is a significant advance in network tomography as it proposes and empirically evaluates the first method that infers the {\em complete} network topology, rather than just logical routing trees from sources.  
\end{abstract}

%% file: introduction.tex
\section{Introduction}
Enterprises rely heavily on the Internet and other communication networks for their operations. However, they lack explicit knowledge about the topological properties of their network, such as the nodes and links of the network and the routes that packets take between their hosts.  In fact, communication networks are often administered by multiple entities and no single entity may have apriori knowledge of the topology of the entire network. However, there are great benefits for enterprises to know the topological properties of their communication network.  For instance, by deducing the graph structure of the network and the routing paths between their hosts, the enterprise can better understand the impact of node (i.e., router) and link failures on their mission-critical communication, leading to better  disaster planning and recovery.  Further, knowing the network topology and routes allow for better performance monitoring and network resource management for enterprise communication. 

 Formally, a communication network $N=(G,H,P)$ can be represented by a graph $G = (V,E)$, where $V$ is the set of nodes and $E$ is the set of links, a set of hosts $H \subset V$, and a set of routing paths $P$ in $G$ between host pairs in $H \times H$. An example of a communication network is shown Figure~\ref{fig:inferenceexample}.  Our work focuses on the problem of infering network $N$ using {\em only} metrics measured from the hosts $H$.  {\em Network inference} includes both {\em topology inference} that infers $G$ and {\em route inference} that infers the routes $P$ between each pair of hosts.

\subsection{Prior Work in Network Inference}
To set our research in context, we review prior work in network inference that has been an active area of research for more than two decades, given its importance in many practical contexts. Much prior work can be put into two broad categories depending on what measurements can be made and to what extent the non-host nodes in $V \setminus H$ assist in those measurements.  

The first category of work assumes that active probes (such as  traceroute and mtrace) and data feeds  (such as BGP) can be used for network inference.  For example, Skitter \cite{Claffy} and its successor Archipelago \cite{CAIDA2015} derives the topology of the Internet using traceroutes and BGP tables.  Rocketfuel \cite{Spring2004} infers the topology of an ISP using traceroutes, BGP, and DNS measurements. Dimes \cite{Shavitt2005} aims to infer topology by running traceroutes from applications installed by volunteers on their personal computers, as opposed dedicated machines as hosts. Doubletree \cite{Donnet2004} modifies traceroute to be more efficient by making the assumption that paths from a source or paths to a destination form a tree.

The above work require the non-host nodes in the network to support specific types of active probes (e.g, traceroutes) and/or to provide measurement feeds (e.g., BGP). However, for reasons of security, many networks (e.g., military networks) do not allow probes such as traceroutes and do not expose other internal network data that may be used for inference.  Even in civilian networks such as the Internet, an increasing fraction of routers do not respond to  traceroutes \cite{Yao2003, Gunes2009}.  Further, future networks may obfuscate topology inference by returning false traceroutes  \cite{Trassare2013}. 

In the past two decades,  such considerations have led to a second category of work  that we refer to as {\em ``network tomography''}  that aims to infer  topology and routes with {\em minimal co-operation} from the non-host network elements (i.e., no traceroutes or data feeds) \cite{vardi1996network,Castro2004,Bu2002,Duffield2002,presti2002multicast,caceres1999multicast}.  In the network tomography literature, network inference is typically performed using only easily-measurable {\em path metrics} derived from the hosts, such as path distance (in number of hops) between hosts and path sharing that is the (relative) amount of link sharing between two host-to-host paths  (c.f., Section~\ref{sec:pathmetrics}). {\em Our focus is network tomography, as we use only path metrics and assume no co-operation from non-host network elements.}

 Early work on network tomography focused in inferring the logical source tree rooted at a host, not the entire network \cite{Ratnasamy1999, Duffield2002,Duffield2004, Bu2002,presti2002multicast}. A source tree is the logical tree formed by the routes from a host as the root to the other hosts as the leaves. A source tree is {\em not a subgraph} of the topology $G$, but rather the logical tree that describes how paths from a source to other destinations  bifurcate. For instance, the source tree rooted at host $A$ in Figure~\ref{fig:inferenceexample}  is shown in Figure~\ref{fig:srctreeA}. Ratnasamy et al. \cite{Ratnasamy1999} propose a method to infer a  binary source tree by using multicast probes sent from a source host to a set of destination hosts. The tree is constructed by observing that destinations experiencing correlated losses have a common shared path from the source, and the amount of correlation increases with the length of the shared path. Subsequently, these results were extended to infer non-binary source trees with theoretical guarantees \cite{Duffield2002}, using delay covariances at the hosts instead of losses \cite{Duffield2004},  and using a train of unicast packets instead multicast \cite{Duffield2006}.
%\textcolor{red}{R. Bowden and D Veitch - Under link spatial dependence \cite{Bowden2018}}.

Besides inferring single source trees in isolation, how multiple source trees intersect has been studied.  Given two source trees, \cite{Rabbat2004,Coates2003} discover the links where the trees intersect, using path sharing metrics between pairs of sources and destinations. But, the technique does not allow the trees to be merged into a single network, unless the trees overlap in very specific ways that do not hold for general topologies. Thus, these techniques cannot be used to produce the complete topology and routes as we do in our work.

Recently, there has been some progress on complete topology inference. A sparse random graph with shortest-path routing can be inferred  with small error and with high probability, though using primitives with no established techniques for measurement \cite{Anandkumar2011}.  Concurrent to our work,  an interesting theoretical advance shows that certain classes of graph topologies can be inferred using stronger primitives that allow the measurement of distances from hosts to certain (non-host) internal nodes \cite{Berkolaiko2018}. Specifically, it is assumed that given paths from a host A to two hosts B and C, the individual distances from A, B, and C to the  internal node where the paths diverge can be measured. However, it is not known how such a stronger primitive can be implemented accurately in a real-world network, while our work uses only path metrics as primitive with well-known accurate implementations.  

{\em Thus, our work is a significant step forward in network tomography as we provide the  first \underline{empirically-}  \underline{validated} method for inferring the \underline{complete} topology and routes of real-world networks.}
\subsection{Path metric inputs to network inference} 
\label{sec:pathmetrics}
As in much of the network tomography literature, we assume that path metrics of two types are available as inputs to network inference: path sharing metrics (PSMs) and path distance metrics (DMs).  Much is known about how to measure them by sending multicast \cite{Duffield2002, Duffield2004, Bu2002, presti2002multicast,caceres1999multicast} or unicast \cite{Duffield2006} packet probes, and passive measurement that deduces the information from existing traffic flows \cite{eriksson2007learning}. We measure PSMs and DMs in standard ways known in the literature. {\em Our contributions lie not in how these metrics are measured, but on how they can used to perform network inference.} 

{\bf 1) Path Sharing Metrics (PSMs).} PSMs measure to what extent routes (i.e., paths) between hosts share links. Let $PSM(S,T_i,T_j)$ represent the number links shared between the paths from a single source host $S$ to two destination hosts $T_i$ and $T_j$. Our work does not require measuring absolute values for the PSMs, but only relative ones. For instance, given a source $S$ and three destinations $T_1$, $T_2$, and $T_3$, it suffices to measure how $PSM(S,T_1.T_2)$ compares with $PSM(S,T_2,T_3)$.
It is well-known how relative PSMs can be computed using latency and/or loss experienced by packet probes from the source host \cite{Duffield2002,Duffield2004, Bu2002,presti2002multicast}. For instance, in Figure~\ref{fig:inferenceexample}, by sending multicast (or, a train of unicast) packet probes from a source $A$ to receivers $C$, $D$, and $E$,  one can infer that $PSM(A,C,E) < PSM(A,C,D)$, since more correlation is expected between the packets received at $C$ and $D$ than between $C$ and $E$.  

In fact, it is well-known from prior work \cite{Duffield2002, Duffield2004} how the entire source tree can be inferred by repeatedly using the relative values of the PSMs. For instance, the source tree in Figure~\ref{fig:srctreeA} can be constructed by making the most correlated pair of destinations (hosts C and D) as siblings. Next, we can extract the next-most correlated destination with either C or D  (host E)  and make E a sibling of the parent of C and D, and so on till the entire source tree is inferred. 

{\bf 2) Distance Metrics (DMs)}.
A DM measures the distance, i.e., the number of links, in the path from a source host $S$ to hosts $T$. Again, we do not require absolute values of the DMs, and relative ones will suffice. More precisely, given a single source $S$ and two destinations $T_1$ and $T_2$, it suffices to measure how $DM(S,T_1)$ compares with $DM(S,T_2)$. For instance, in Figure~\ref{fig:inferenceexample}, it is easy to see that $DM(A,F) < DM(A,D)$.

A standard approach to measuring DMs is to use the  $TTL$ field  \cite{Postel1981} of the IP Header. The source host initializes the $TTL$ value in the IP Header to 255, and each node on the path to the destination decrements the $TTL$ by 1. At the destination, the $TTL$ value is read from the IP Header. And the end-to-end distance between the source and destination host is calculated by taking the difference.

\begin{figure}[!t]
%\vspace*{-0.2in}
\centering
\includegraphics[width=0.75\linewidth, height=38mm]{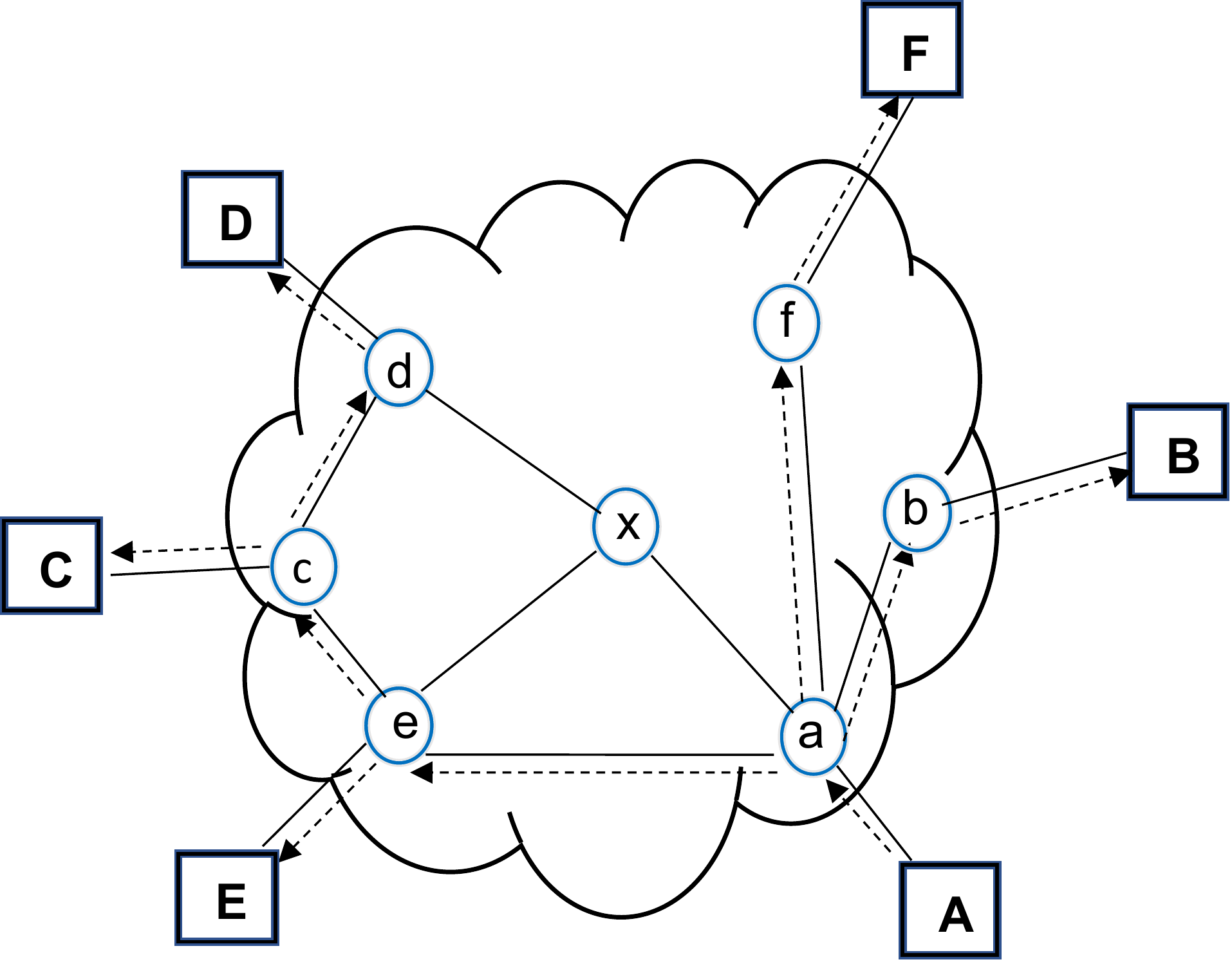}
\caption{Communication network with hosts, internal nodes (routers) and links. The routing paths from host A to other hosts are shown in dotted lines.}
\vspace*{-0.1in}
\label{fig:inferenceexample}
\end{figure}

\begin{figure}[!t]
%\vspace*{-0.2in}
\centering
\includegraphics[width=0.5\linewidth, height=46mm]{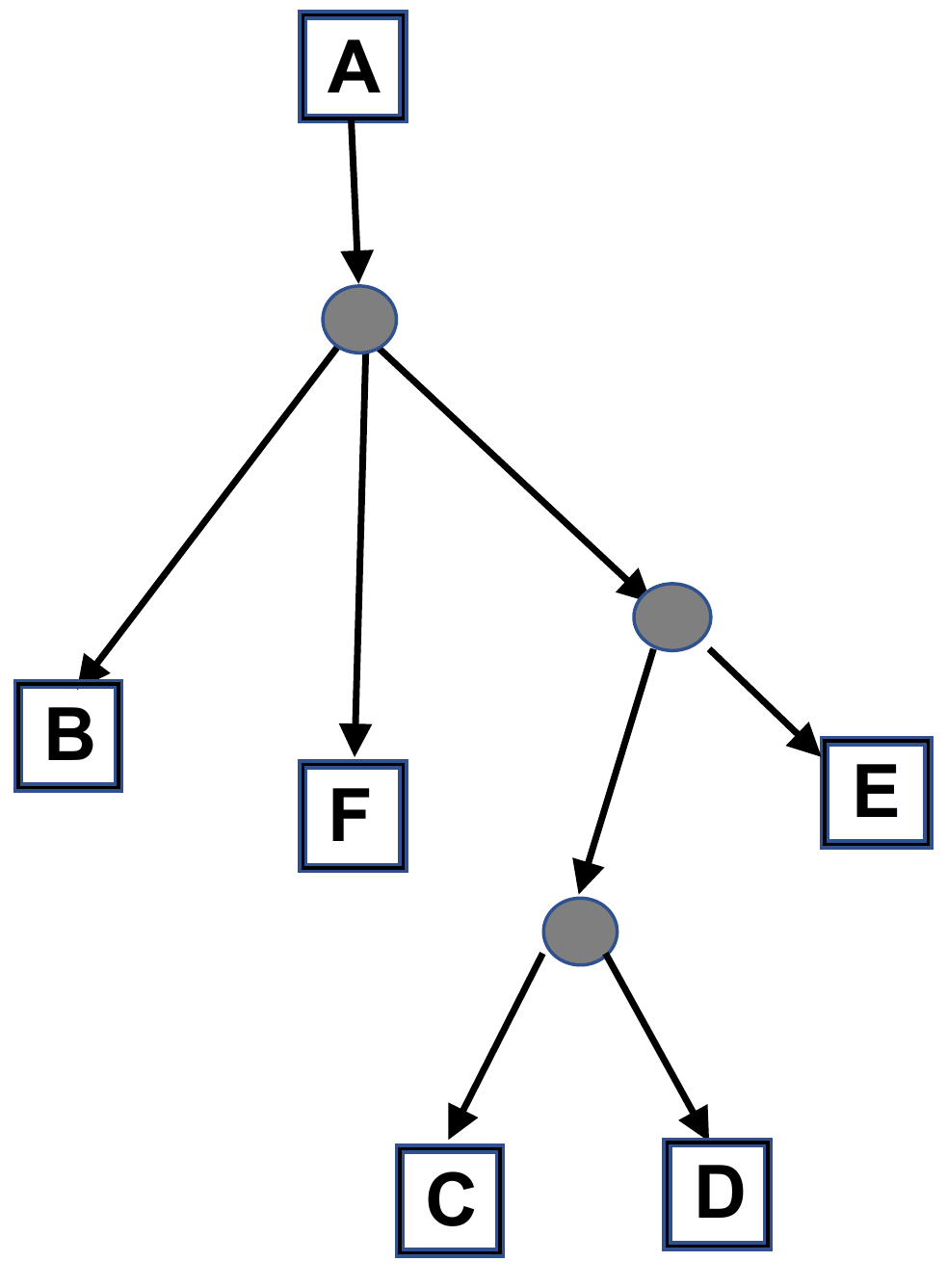}
\caption{Source tree rooted at host A}
\label{fig:srctreeA}
\vspace*{-0.15in}
\end{figure}

\subsection{Our Contributions} 
To our knowledge, our work is the first to propose, implement, and empirically validate a method for inferring the {\em complete} network topology and routing paths of  a communication network, using only path metrics. Prior work on network tomography has been limited to inferring individual source trees, how these  trees intersect, or use stronger measurement primitives that are not easily implementable. Specific contributions follow.

1) We propose a novel theoretical approach (OCCAM) that applies the Occam's razor principle \cite{Occam} to pose and solve an optimization problem to find the ``simplest'' network that obeys observations. The solution to the optimization problem yields the inferred network and routing paths. Our optimization approach is a new way of thinking about network inference and contrasts with other statistical ways of thinking about the problem known in the prior art, such as Maximum Likelihood Estimation (MLE) \cite{Duffield2002}. We prove the correctness of OCCAM by formally showing that it provides a solution that satisfies all PSM and  DM observations. 

2) We evaluate OCCAM on several real-world ISP network topologies and show that it provides high-quality inferences that agree closely with ground truth. The average network similarity score of the inferred topology with respect to ground truth is 93\%. The inferred routing paths have a small average edit distance of 0.20 from ground truth.

3) We analyzed the robustness of OCCAM's inference when a fraction of its inputs are erroneous, as would be the case if its PSM and DM inputs are derived from actual network measurements. For the networks we tested, OCCAM produced a high-quality  inference even when a random 20-30\%  of the PSMs and DMs  were erroneous.

4) To more closely simulate real-world network inference, we implemented multiple ISP topologies on the  DETER \cite{Mirkovic10the} network emulator. Using unicast packet probes, we derive PSM and DM values from packet-level measurements from DETER. Using these measured values as inputs to OCCAM, we show that it produces high-quality inferences, close to ground truth. 

5) It is well-known from prior work that source trees can be constructed with PSM inputs. It is natural to ask if the source trees  so constructed can be ``stitched'' together to create the complete network topology. This yields a variant of network inference where you are provided source trees (instead of PSMs) as the measured input, in addition to DM constraints. We show that OCCAM's optimization can be modified to perform tree stitching. In this variant, OCCAM performed similar to the original version when PSMs and DMs are provided as the measured input.

%% file: theory.tex
\section{The OCCAM approach}
\label{sec:theory}
The OCCAM approach infers a network $N'=(G',H,P')$ using as inputs measurements from the actual ``ground-truth'' network $N=(G,H,P)$ as follows.

\noindent {\em Measurement Inputs.}  DM metrics and relative values of the PSM metrics are measured from the hosts $H$ of $N$ using standard techniques described in Section~\ref{sec:pathmetrics} and provided as inputs. 
\begin{enumerate}
\item {\em Optimization Step.}  Network inference is formulated and solved as an optimization problem where the ``simplest'' network satisfying the observed PSM and DM constraints is produced as a solution.
\item {\em Inference Step.} The inferred network $N'=(G',H,P')$ is constructed from the solution of the optimization. 
\end{enumerate}

\subsection{The Optimization Step}
\label{subsec:opt}
The key idea of our approach is to view network inference as an optimization problem where the ``simplest'' network satisfying the observed PSM and DM constraints is produced as a solution. We view this approach as analogous to Occam's razor that is a heuristic element of the scientific method and advocates the construction of the simplest and most parsimonious model that obeys the empirical observations. We capture the existence of nodes and links, as well as the membership of links in routing paths, as indicator variables whose values are set by the optimization process\footnote{We do not know how many non-host nodes exist apriori. So, we define indicator variables for an upper bound on the number of non-host nodes and allow the optimization to decide how many such nodes actually exist by setting those indicator variables.}  (cf. Table~\ref{table:variables}). The values of these variables as set by the optimization yield the inferred network topology and routing paths.

\textbf{Objective function.} There are many notions of simplicity possible in a network setting. We use the notion that the inferred network should have (i) the smallest number of links, and (ii) the smallest total host-to-host shortest path distance. We can express the notion of simplicity as the objective function that needs to be minimized as follows.
\vspace*{-0.05in}
\begin{equation}
\label{eq:objective}
\begin{split}
\min: \;\; \alpha \; \sum_{S \in H} \sum_{T \in H} m_T^S + \;(1- \alpha) \; \sum_{i \in V} \sum_{j \in V} w_{ij},
\end{split}
\end{equation}
where $m_T^S$ is an integer variable denoting the length of the path from source host $S \in H$ to destination host $T \in H$, $w_{ij}$ is a variable indicating if a link exists between node $i$ and node $j$ in the inferred network, and $0 \leq \alpha \leq 1$ weighs the relative importance of the two components of the objective function.

\textbf{Path sharing.} For each measured relative PSM metric of the form  $PSM( S, T_1, T_2) <   PSM( S,T_2, T_3)$  we add the constraint below.

\begin{equation}\label{eq:entanglement}
\sum_{i \in V} v_{i}^{S,T_1} v_{i}^{S,T_2} < \sum_{i \in V} v_{i}^{S,T_2} v_{i}^{S,T_3},
\end{equation}
where  $v_{i}^{S,T}$ is a variable indicating if node $i$ is on the path from host $S$ to host $T$. The \textit{LHS} of the above inequality thus counts the number of nodes 
in the intersection of  paths  from $S$ to $T_1$ and $T_2$.  Similarly the \textit{RHS}  counts the number of nodes present in the intersection of paths from $S$ to $T_2$ and $T_3$.  

\textbf{Distance metrics.}
For each measured DM metric, if $DM(S,T_1) < DM(S,T_2)$, we add the constraint below.
\begin{equation}\label{eq:relative}
m_{S}^{T_1} < m_{S}^{T_2}
\end{equation}
where $m_{S}^{T}$ is an integer variable indicating the distance, in terms of number of links on the path from $S$ to $T$. When the absolute value of the DMs can be calculated accurately, the above constraint can be replaced by
\begin{equation}\label{eq:accurate-distances}
m_{S}^{T} = DM(S,T),
\end{equation}
for every pair of hosts $(S,T)$. 
However, in practice, we have observed that using the constraint above on absolute DM values increases the run time of the algorithm. So, in our experiments, we run OCCAM with the weaker constraint on relative DM values of Equation~\ref{eq:relative} and empirically show that it is sufficient to obtain a high-quality network inference.

%\footnote{\textcolor{red}{When DMs can be calculated directly, the above constraint can be replaced by setting $m_{S}^{T} = DM(S,T)$ for every pair of hosts $(S,T)$. However, in practice, we have observed that it increases the runtime of the algorithm.}}

\textbf{Source tree property.} Let $P^S \subset P$ be the set of routing paths in $G$ between host pairs $\{S\} \times H$; we add the following constraints to ensure that links belonging to $P^S$ form a tree.  
\begin{equation}\label{eq:src-tree}
\sum_i s_{i,j}^{S} \leq 1 \quad \quad  \forall j \in V  \quad S \in H
\end{equation}
where $s_{i,j}^{S}$  is a variable indicating if link $(i,j)$ is on any of the paths in $P^S$. The constraint ensures that for every node $j \in V$, the number of links in $P^S$ that terminate at node $j$, is at most 1. Thus, it ensures that there is at most one unique path to node $j$ from source host $S$.
 
\textbf{Source-oblivous paths.} Typically, a packet at a node $i \in V$ is forwarded to the next node $j \in V$ by consulting a routing table that provides the ``next-hop'' for each destination $T \in H$, independent of the packet's source. In particular, two packets arrive at a node $i$ from different sources are forwarded to the same next node $j$ if they are going to the same destination $T$.  We capture this as follows:

\begin{equation}\label{eq:dst-tree}
\sum_{j \in V} d_{i,j}^T \leq 1 \quad \quad  \forall i \in V,  \quad T \in H, 
\end{equation}
where  $d_{i,j}^T$ is an indicator variable indicating if a link $(i,j)$ is on any of the paths $P_T$, where $P_T \subset P$ is a set of routing paths in $G$ between host pairs in $H \times \{T\}$. Above equation ensures that, if a node $i$ is on any of paths to destination $T$, the number of possible forward hops is at most 1.

\begin{figure}
  \centering
  \begin{tabular}{|r|l|}\hline%                                                                                                                                                                                                          
    \bfseries Symbol & \bfseries Meaning \\ \hline  \hline                                                                                                                                                                                      
    $s_{i,j}^S$ & An indicator variable indicating \\ 
                & if link $(i,j)$ belongs to any path with host \\
                & $S$ as the source.  \\ \hline
    $d_{i,j}^T$ & An indicator variable indicating if \\ 
                & link $(i,j)$ belongs to any  path with  \\
                & host $T$ as the destination \\ \hline
    $m_j^S$ & An integer variable denoting the \\ 
            & number of hops required to reach \\
            & node $j$ from host $S$ \\ \hline
    $v_j^{S,T}$ & An indicator variable indicating if \\
                &  node $j$ is on the path from host \\
                &  $S$ to enclave $T$ \\ \hline   
     $w_{ij}$ & an indicator variable indicating if the \\ 
     		& link $(i,j)$ is present in the 
		inferred graph \\
		\hline
  \end{tabular}
  \caption{Output variables set by optimization}
  \label{table:variables}
  \vspace*{-0.15in}
\end{figure}

\textbf{Populating the $d_{i,j}^T$ variables.} Link $(i,j)$ is in $P_T$ if only if there exists a source $S$ such that both of the following hold.
\begin{enumerate}
\item Link $(i,j)$ belongs to $P^S$, i.e., $s_{i,j}^S$ is  $1$.
\item Node $j$ is on the path from $S$ to $T$, i.e., $v_j^{S,T}$ is  $1$. 
\end{enumerate}
Note that the above two conditions imply that link $(i,j)$ is in $P_T$ because the first condition implies that the path from $S$ to $j$ must go through $(i,j)$. Thus, $d_{i,j}^T$ can be set to $1$, if only  if  $\sum_{S \in H} s_{i,j}^S v_j^{S,T}$ is positive.  This can be expressed using the following constraint.
\begin{equation}\label{eq:dst-tree-populate}
\begin{split}
-M(1 - d_{i,j}^T) < \sum_{S \in H} s_{i,j}^S v_j^{S,T} \leq M d_{i,j}^T  & \quad \forall i,j \in V,
\end{split}
\end{equation}
where $M$ is a suitably large constant. Note that if $\sum_{S \in H} s_{i,j}^S v_j^{S,T} $ is zero, the first inequality above forces $d_{i,j}^T$ to be zero. Else, if $\sum_{S \in H} s_{i,j}^S v_j^{S,T} $ is positive, the second inequality above makes $d_{i,j}^T$  to be $1$.

\textbf{Constraints to calculate distances.} 
\begin{equation}\label{eq:distance}
 m_{j}^{S} = \sum_{i \in V} s_{i,j}^{S} (m_{i}^{S} + 1) \quad \forall S\in H \quad j \in V,
\end{equation}
where $m_{j}^{S}$ is the number of hops to node $j$ from source $S$. The above constraint evaluates variable $m_{j}^{S}$ by stating that if there exists an incoming link $(i,j)$ in $P^S$, i.e., if $s_{i,j}^{S} = 1$, then the value of $m_{j}^{S}$ can be computed as $m_{i}^{S} + 1$. (Note that constraint (\ref{eq:src-tree}) ensures that there is at most one such link $(i,j)$ in  $P^S$).  If there is no such link $(i,j)$ in $P^S$, i.e. if $s_{i,j}^S = 0$, then $m_{j}^{S}$ is $0$. In this case, we say that node $j$ is not on any paths in $P^S$. We initialize the variable $m_{S}^{S}$ to $0$. 

\textbf{Tracing a host-to-host path.} We add below constraints to find nodes that are on the path from host $S$ to host $T$.  These variables are used in (\ref{eq:entanglement}) to encode PSM constraints. To determine if node $i$ is on the path from $S$ to $T$, we add the following constraint.
\begin{equation}\label{eq:node-on-path}
v_{i}^{S,T} = \sum_{j \in V} v_j^{S,T} s_{i,j}^S \quad \forall S,T \in H, \forall i \in V - H
\end{equation}
Node $i$ is on the path from $S$ to $T$, if there exists a node $j \in V$ such that (i) $j$ is on the path from $S$ to $T$, i.e., ($v_j^{S,T} = 1$) and, (ii)  there exists an outgoing link from node $i$ to node $j$ in $P^{S}$. The above two conditions suffice because (\ref{eq:src-tree}) ensures that there can be at most one incoming link to $j$ in $P^S$, and if such a link exists, $i$ should necessarily be on the path from $S$ to $T$.

\textbf{Boundary conditions.} Paths in $P^S$ should always contain a outgoing link from source $S$, and an incoming link at each destination host $T \in H \setminus S$. We add the following constraints for each host $S \in H$ to ensure that the above requirement is met.
\begin{equation}\label{eq:boundary-presence-src}
\sum_{j \in V} s_{S, j}^S = 1,
\end{equation}
\begin{equation}\label{eq:boundary-presence-dst}
\sum_{j \in V} s_{j, T}^S = 1 \quad \forall T \in H \setminus S,
\end{equation}
Similarly, we need to ensure that paths in $P^S$ have no incoming link at source $S$ and no outgoing link at a destination host $T \in H \setminus S$. We add the following constraints for each host $S \in H$ to ensure that the above requirement is met.
\begin{equation}\label{eq:boundary-absence-src}
\sum_{j \in V} s_{j, S}^S = 0,
\end{equation}
\begin{equation}\label{eq:boundary-absence-dst}
\sum_{j \in V} s_{T, j}^S = 0 \quad \forall T \in H \setminus S,
\end{equation}

We also ensure that if there is an outgoing link $(j,k)$ at node $j$ in $P^S$, then there must exist an incoming link at node $j$ that is in $P^S$. 
\begin{equation}\label{eq:boundary-path-constraint} 
s_{j,k}^S \leq \sum_{i \in V} s_{i,j}^S \quad \forall j \in V \setminus S,
\end{equation}
The above equation says that if $s_{j,k}^S$ is 1, i.e.,  there exists an outgoing link $(j, k)$ at node $j$, then the term $\sum_{i \in V} s_{i,j}^S$ cannot be $0$, i.e., there must exist an incoming link at node $j$. Note that we do not write the constraint if node $j$ is the source host $S$. 

\textbf{Dealing with inaccurate measurements.} 
The PSM and DM metrics derived as inputs can sometimes be inaccurate in real world scenarios. For instance, the PSMs can be inaccurate in real-world networks when there is no multicast available, a train of unicast packet probes must be used, and there is a significant amount of background traffic. Such was the case with some of our experiments on the DETER testbed. An  variant of the optimization step that we used to tackle measurement inaccuracies is to convert the hard PSM and DM constraints in (\ref{eq:entanglement}) and (\ref{eq:relative}) into soft constraints by moving them to the objective function. That is, we add a third component to the objective function in (\ref{eq:objective}) that represents the number of PSM and DM constraints that are violated. Thus, violations of PSM and DM constraints are minimized, along with other considerations. Thus, the new objective function finds the ``simplest'' network that obeys ``most'' of the observed measurements.  As we show later, with this approach, OCCAM made accurate network inferences even when 20-30\% of the PSM constraints were incorrect.

\textbf{Discussion.} The OCCAM approach assumes that the network is ``simple'' in two different respects. First, it posits that the network  is itself ``simple'' in the sense of having the fewest number of links and the shortest distances between the hosts. This is reflected in the objective function that is minimized. Real-world network designers may not always design networks that strictly obey that notion of simplicity.  Second, OCCAM assumes that the paths on the network are also ``simple''  and posits that paths are source-oblivious and satisfy the source tree property.  Again, specific real-world routing protocols may disobey some of these properties some of the time.  

In the philosophy of science, Occam's Razor is used as an aesthetic principle for choosing the simplest theory that fits the observations. Its use has been unreasonably effective in producing sound scientific principles. Likewise, as we show in Section~\ref{sec:empirical}, even though OCCAM's quest for simplicity can produce erroneous results, it generally results in high-quality network inference. We also observe that even when some individual assumptions are violated, OCCAM can correct for the erroneous assumptions and still produce a high-quality network inference.
\vspace*{-0.05in}
\subsubsection{Solution approach} The optimization problem formed with the objective function in Equation~\ref{eq:objective} and constraints that include Equations~\ref{eq:entanglement}, ~\ref{eq:relative} and ~\ref{eq:src-tree} to \ref{eq:boundary-path-constraint} is a  \textbf{M}ixed \textbf{I}nteger \textbf{B}ilinear \textbf{P}rogram (\textit{MIBP}). For ease of using  solvers such as CPLEX, we  linearized the problem to form a \textbf{M}ixed \textbf{I}nteger \textbf{P}rogram (\textit{MIP}) as follows. 

\textbf{Linearize a product of binary variables.} Note that the constraints in Equations \ref{eq:entanglement}, \ref{eq:node-on-path} and \ref{eq:dst-tree-populate} have bilinear terms that are a product of two binary variables. We linearize each such bilinear term as follows. Consider a bilinear term of the form $xy$, where $x$ and $y$ are binary variables. Replace the term $xy$ with a new binary variable $z$ and add the three constraints: $z \leq x$, $z \leq y$ and $z \geq x + y - 1$.  The first two inequalities ensure that $z$ is $0$, if either $x$ or $y$ is $0$. The third inequality ensures that $z$ is $1$ if both $x$ and $y$ are $1$.

\textbf{Linearize the product of an integer and a binary variable.} Note that the constraints in Equation \ref{eq:distance} have bilinear terms that are a product of an integer and a binary variable. Suppose that a bilinear term has the form $ib$, where $b$ is a binary variable and  $i$ is an integer variable lower bounded by $0$ and upper bounded by  $I$. The product term $ib$ can be linearized as follows. Replace the term $ib$ with a new integer variable $z$  and add four constraints: $z \leq I b$, $z \leq i$, $z \geq i - (1 - b)I$,  and $z \geq i$. Note that if $b$ is zero, than the first inequality ensures that $z$ will be zero as well (note that the third inequality only states that z has to be greater than a negative number). On the other hand, if $b$ is 1, the first two inequalities ensures that $z \leq i$. The third and fourth inequalities ensure that $z \geq i$. Together, this ensures that $z$ equals $i$.

{\bf Using the CPLEX solver.} We use the distributed parallel MIP feature of CPLEX to solve our problem on a server cluster.  We set the relative MIP gap to 0.15, which means that CPLEX stops looking for solutions once it finds one within 15\% of the optimal. Empirically, for networks that we evaluate in this paper, we have found that a MIP gap of 0.15 produces solutions that are reasonably accurate within a run time that does not exceed 10 to 15 minutes.

\begin{algorithm}\label{al:network-construct1}
\SetAlgoLined
\caption{\textit{GRAPH-CONSTRUCT-I}}
\textbf{Input :} Set of hosts H and solution variables \;
\textbf{Output :} Network $N'=(G',H,P')$. \;
\BlankLine
Initialize: $V' \leftarrow \phi$, $E' \leftarrow \phi$, $P' \leftarrow \phi$ \;
\BlankLine
\ForEach{$(S,T) \in H \times H$} {%
	\BlankLine
	$i_0 \leftarrow T$, $k = 0$ \;
	\BlankLine
	\While{$i_k \neq S$}{
		\BlankLine
		$k \leftarrow k + 1$ \;		
		\BlankLine
		Find a node $i_k$ such that $s_{i_k, i_{k-1}}^S$ equals $1$ \; 		
	}
	$\pi(S,T) = \{ i_k, i_{k-1}, i_{k-2}, ....... , i_1, i_0 \}$ \;
	$V' = \bigcup_{k} i_k$ \;
	$E' = \bigcup_{k, k-1} (i_{k}, i_{k-1})$ \;
}
$G' \leftarrow (V', E')$ \;
$P' = \bigcup_{(S,T) \in H \times H} \pi(S,T)$ \;
\Return $N'=(G',H,P')$ \;
\end{algorithm}

\vspace*{-0.1in}

\subsection{Inference Step}
Algorithm GRAPH-CONSTRUCT-1 infers a network $N' = (G', H, P')$ using the values set to the $s_{i,j}^S$ variables in the optimization step. For a fixed source-destination pair $(S,T) \in H \times H$, a routing path $\pi(S,T)$ is inferred as follows. Starting from $T$, the while loop in lines 6-9 iteratively finds nodes to build a path towards source $S$. In each iteration $k$, a node $i_k$ is found such that $s_{i_{k}, i_{k-1}}^S$ equals 1, and the loop terminates when $i_k$ is the source host $S$. The path $\pi(S,T)$ is then constructed as the union of the links $\bigcup_{k} (i_{k}, i_{k-1})$. Each node $i_k$ and link $(i_k, i_{k-1})$ is added to graph $G'$ in Lines 11 and 12 respectively. Thus at the end of the for loop in line 13, the routing paths $P'$ and graph $G'$ is constructed.

\begin{comment}
The inferred network $N' = (G', H, P')$ can be determined from the values of the variables set by the optimization (cf. Table~\ref{table:variables}) using Algorithm \ref{al:network-construct1}. \hi{short description}
\textcolor{red}{The inferred graph $G = (V,E)$ is constructed as follows. We include a link $(i,j)$ in $E$  if and only if $s_{i,j}^S$ is set to 1, for some  $S\in H$. Further, a node $i$ is in $V$ if and only if $i$ has an incoming or outgoing edge in $E$. Finally,  to determine the set of routing paths $P$, the nodes on each routing path $\pi(S,T)$ is simply all $j \in V$ for which $v_j^{S,T}$ is set to 1. To sequence the nodes on the path $\pi(S,T)$, we can sort the nodes on the path using the distance from the source provided by $m_j^S$ variables.} 
\end{comment}

%We ran the optimization across 5 machines, each with 8 cores and 16GB RAM.

\subsection{Correctness of OCCAM}
\vspace*{-0.06in}
\begin{theorem} \label{th:psm} Given PSMs and DMs as inputs, OCCAM infers a network $N'=(G',H,P')$; such that the routing paths P' satisfy the following properties:
\begin{enumerate}
\item The set of routing paths $P'$ contains an unique acyclic path between each pair of hosts; and
\item $G'$ and $P'$ satisfies all the given PSM and DM constraints.
\end{enumerate} 
\end{theorem}
\begin{pfsketch}
We first show that OCCAM infers an unique acyclic routing path $\pi(S,T)$ between every pair of hosts $(S,T) \in H \times H$. Constraints in (\ref{eq:boundary-presence-dst}) ensures there exists a link $(j,T)$, for some $j \in V$, such that $s_{j,T}^S$ equals 1. Link $(j,T)$ is on the path $\pi(S,T)$. Now constraints in (\ref{eq:boundary-path-constraint}) ensures that for link $(j,k)$, if $s_{j,k}^S = 1$, then there exists an incoming link $(i,j)$ such that $s_{i,j}^S =1$, unless  $j$ is the source host $S$. Link $(i,j)$ is in $\pi(S,T)$. Thus link $(j, T)$ triggers the formation of a path begins at source $S$ and terminates at $T$. The path is acyclic as it would otherwise violate constraints in (\ref{eq:distance}). Now constraints in \ref{eq:distance} ensure $m_{S}^{T}$ equals the length of path $\pi(S,T)$ and $v_{i}^{S,T}$ equals $1$ only if node $i$ is on the path $\pi(S,T)$. Thus constraints in \ref{eq:entanglement} and \ref{eq:relative} ensure the PSM and DM constraints are satisfied. We provide the complete proof in Appendix ~\ref{app:correctness}.
\end{pfsketch}

The above shows that the output $N'$ of OCCAM obeys all the PSM and DM constraints, but it is theoretically possible that there are other optimal solutions that are different from $N'$. Further, it is also possible that the ground truth differs from $N'$ because it may not be a network that minimizes the objective function. We empirically show in Section~\ref{sec:empirical} that OCCAM produces a network that is very similar to the ground truth, though always not the same. However, for specific classes of networks, OCCAM provably produces the ground truth. We show below  that if the ground truth network is a tree then there is exactly one optimal solution and OCCAM's output exactly corresponds to the ground truth network. 

\begin{theorem} Let the PSMs and DMs be derived from a ground truth network  $N$ that is a tree.  Given the PSMs and DMs as input, OCCAM's output is the ground truth network $N$.
\end{theorem}
\begin{proof} Our proof builds on the main theorem of Hakimi and Yau  \cite{hakimi}. Given a graph, its distance matrix $D$ provides the shortest distance between each pair of external nodes (i.e., hosts) in the graph. {\em Theorem 6} in \cite{hakimi} shows that if there exists a tree that satisfies a distance matrix $D$, then the graph with smallest number links that satisfies $D$ is unique and equals the tree. 

Suppose that our ground truth network $N$ is a tree. Further, suppose that we run OCCAM with its objective function set to minimizing the total number of links (by setting $\alpha = 0$ in Equation~\ref{eq:objective}) and with the absolute distance constraint (Equation \ref{eq:accurate-distances}). OCCAM outputs the graph with the smallest number of links in the feasible region defined by its constraints.  Since $N$ is the ground truth, $N$  satisfies all constraints considered by OCCAM, including all the absolute distance constraints, i.e., $N$ is in the feasible region. From  Hakimi and Yau, we know that the graph with the minimum number of links in the feasible region is unique and, hence, must equal $N$. So, OCCAM correctly outputs $N$.
\end{proof}
%\subsection{Hardness}
%\begin{theorem} The problem of inferring a network  $N'=(G',H,P')$ with the least number of links that satisfies the given PSMs and DMs is NP-%Complete.
%\end{theorem}
%\begin{pfsketch}
%\end{pfsketch}

 %We first show that for a source $S \in H$, the tree constructed using the links in $P^S$ satisfy all the PSM and DM constraints. We will show that $P^S$ is unique and the links in $P^S$ would have a one-one correspondence with the links in $G$.  This shows that  To show that OCCAM infers no more links, we will show that links in $P^S$ are necessary to satisfy all the DM constraints and are also sufficient to satisfy the PSM constraints. Thus, any feasible network should contain at least all the links in $P^S$. And as OCCAM infers a network that contains minimum number of links which satisfies the PSMs and DMs, the inferred network contains only the links in $P^S$. Thus OCCAM infers a network that satisfies all the constraints and is equivalent to the ground truth network $N=(G,H,P)$. We provide complete proof in Appendix~\ref{app:correctness}.

%% file: empirical.tex
\section{Empirical results}
\label{sec:empirical}
We use several real-world  networks obtained from topology-zoo \cite{Knight2011}  to evaluate OCCAM (see Table~\ref{table:networks}). To judge the quality of the network inference produced by OCCAM, in Section~\ref{sec:quality}, we develop metrics that can compare two networks and quantify its similarity.  Later, we outline two types of experiments and results.

\begin{figure}[!htb]
\centering
\includegraphics[width=0.8\linewidth, height=40mm]{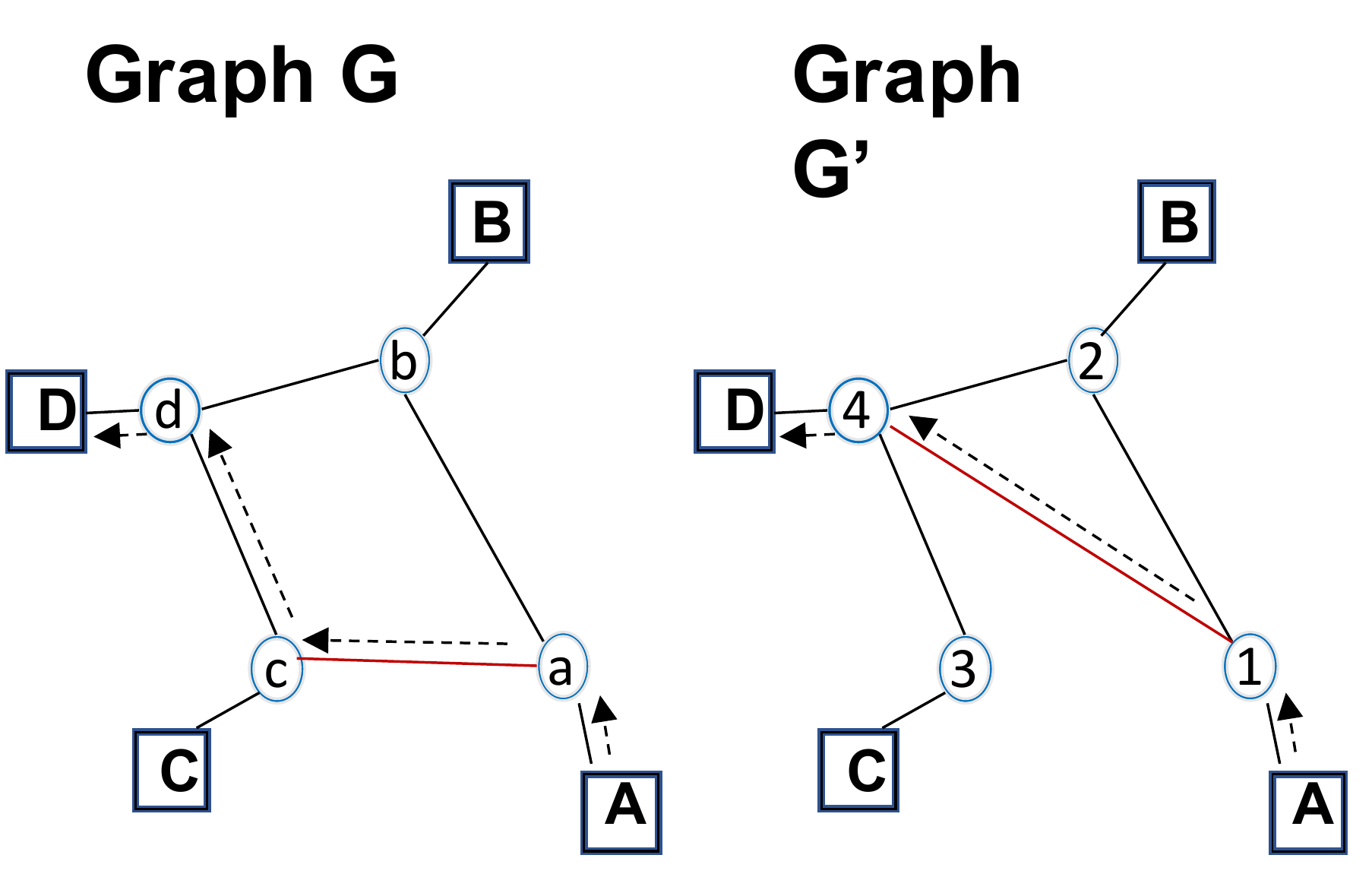}
\caption{Graph example to illustrate NS score and PED}
\label{fig:metricexample}
\vspace*{-0.2in}
\end{figure}

\begin{figure}[!htb]
  \centering
  \begin{tabular}{|r|l|}\hline%                                                                                                                                                                                                          
    \bfseries Topology & \bfseries Description \\ \hline  \hline                                                                                                                                                                                      
   ATT & Backbone network of a major  \\ 
   	& US ISP. \\ \hline
   Tata & Backbone network of a major \\
   	&  Indian ISP \\ \hline
   Bandcon & Content delivery service provider \\ \hline
   Colt & A network providing high \\ 
   	& bandwidth and voice services \\
   	& Europe, Asia and North America. \\ \hline
   Columbus & TV, telephone and broadband ISP \\ 
   	& in the Caribbean \\ \hline
   Dfn & A popular ISP in Oregon, USA \\ \hline
   Evolink & Widely used ISP in Europe \\ \hline
   Rnp & a nation-wide Internet network \\
   	&  infrastructure for the academic \\ 
	& community at Brazil. \\ \hline
   Sanet & Academic network of  national \\
   	& research and education networking \\
	& organisation of Slovakia \\ \hline
   Sinet & Security innovation network, focused \\
   	& on supporting entrepreneurial \\ 
	& companies that \\
	& build cybersecurity solutions. \\ \hline
   Surfnet & SURF an Internet Provider that offers \\ 
   	& students, lecturers and scientists in the \\
	& Netherlands \\ \hline
   6et1 &A hand crafted topology used in the initial \\
   	& stages to test OCCAM \\ \hline
   \end{tabular}
  \caption{A sub-network of the above real-world networks were used for evaluating OCCAM.}
  \label{table:networks}
\end{figure}

\subsection{Quality metrics for network inference} 
\label{sec:quality}

Given a communication network $N=(G,H,P)$ and an inferred network $N'=(G',H,P')$, we introduce two metrics below that quantitatively measure the quality of inference.

\subsubsection{Network Similarity (NS)} The {\em NS} score measures how close the inferred graph $G' = (V', E')$ is to the ground truth of $G = (V,E)$. Intuitively, we compute the ``best''  one-to-one mapping $\phi: V \rightarrow V'$ to match the vertices of one graph with the vertices of the other\footnote{Since G and G' have the same hosts, $\phi$ maps hosts in $V$ to the corresponding hosts in $V'$. If $V$ and $V'$ have different sizes, some nodes in the larger set are left unmapped.}. We then compute the percentage of links that are matched under $\phi$, i.e., percentage of links present in both graphs. Formally, 
%$$ NS(G,G') = \max_{\phi: V \rightarrow V'} \left  (\frac{\sum_{i,j \in V \times V} E_{i,j} \wedge E'_{\phi(i),\phi(j)}} {\sum_{i,j \in V \times V} E_{i,j} \vee E'_{\phi(i), \phi(j)}} \right) \times 100,$$
\begin{multline}
NS(G,G') =   \\ \max_{\phi: V \rightarrow V'}  \left  (\frac{100 \times \sum_{i,j \in V \times V} E_{i,j} \wedge E'_{\phi(i),\phi(j)}} {|E| + |E'| - \sum_{i,j \in V \times V} E_{i,j} \wedge E'_{\phi(i), \phi(j)}} \right),
\end{multline}
where $E_{i,j}$ (resp., $E'_{\phi(i),\phi(j)})$ are indicator variables that is set to $1$ if the corresponding link is present in $G$ (resp., $G'$) and $0$ otherwise, $\wedge$ is the boolean AND operator, and $\vee$ is boolean OR operator. Note that the numerator evaluates the number of links that are in common between the graphs and the denominator is the total number of links present in either graph. Note that the when $G$ and $G'$ are identical, the NS score is a 100\%. Where as if $G$ and $G'$ have complimentary links, no links match and the NS score is 0\%. In general, NS score is a measure of network similarity with values between these two extremes.

An example, Figure \ref{fig:metricexample} shows two graphs $G = (V,E)$ and  $G'=(V',E')$. To evaluate $NS(G,G')$, we first find the one-to-one mapping $\phi: V \rightarrow V'$ that maximizes the matched links\footnote{In general, finding $\phi$ to maximize the NS score is itself a computationally hard program that is related to the graph isomorphism problem for which no polynomial time algorithm is known.  However, for our specific evaluations, we exhaustively searched one-to-one mappings, $V \rightarrow V'$, and choose the mapping with the best NS score. Optimizing the evaluation process itself is beyond the scope of our work.
}. In our case, $\phi =  \{(a, 1), (b,2), (c,3), (d,4)\}$.  Under the mapping, we see that all links, except link $(a,d) \in E$ and $(1,3) \in E'$, can be matched.  Thus, the numerator in the NS score that corresponds to the total number of matched links is $7$. And, the denominator in the NS score corresponds to the union of links in $G$ and $G'$ under the mapping $\phi$, which evaluates to $9$. Thus, NS(G,G') is $77\%$.

 \subsubsection{Path Edit Distance} The {\em PED} metric for path sets $P$ and $P'$ is the average path edit distance between the corresponding paths in $P$ and $P'$. Note that given the one-to-one function $\phi$, each path $\pi \in P$ has a corresponding path $\pi' \in P'$ such that the two corresponding paths connect the same host pairs under $\phi$. Path edit distance between two paths $\pi$ and $\pi'$ is simply the number node insertions, deletions and substitutions required to convert one path to the other.  The overall PED is simply the average PED of the individual path pairs.
 
 As an example, we show the PED calculation for the path from host $A$ to host $D$ in Figure \ref{fig:metricexample}. The path $P(A,C)$ in $G$ is $\{A, a, d, c, C\}$, and the path $P'(A,C)$ is $\{A, 1, 3, C\}$. Under the mapping, $\phi \rightarrow V \times V' : \{(a, 1), (b,2), (c,3), (d,4)\}$, the path $P(A,D)$ can be rewritten as $\{A, 1, 4, 3, D\}$, which is at an edit distance of $1$ from $P'$, since a single edit of the deletion of node $4$ is required.

\begin{figure}[!thb]
\centering
\includegraphics[width=0.9\linewidth, height=42mm]{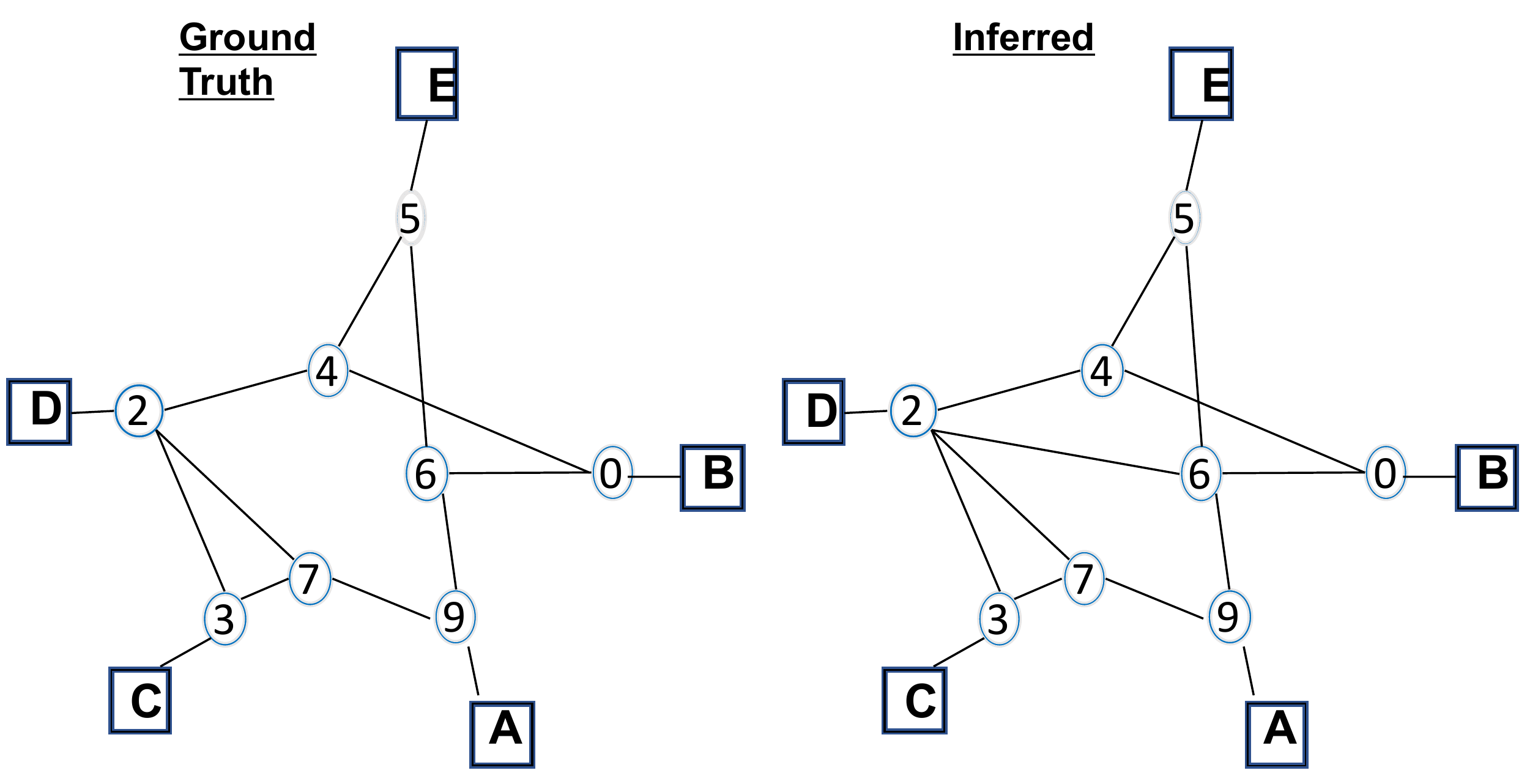}
\caption{AT\&T network inferred with an network similarity of over 93\%.}
\label{fig:att}
\vspace*{-0.1in}
\end{figure}

\begin{figure}[!thb]
\centering
\includegraphics[width=0.95\linewidth, height=50mm]{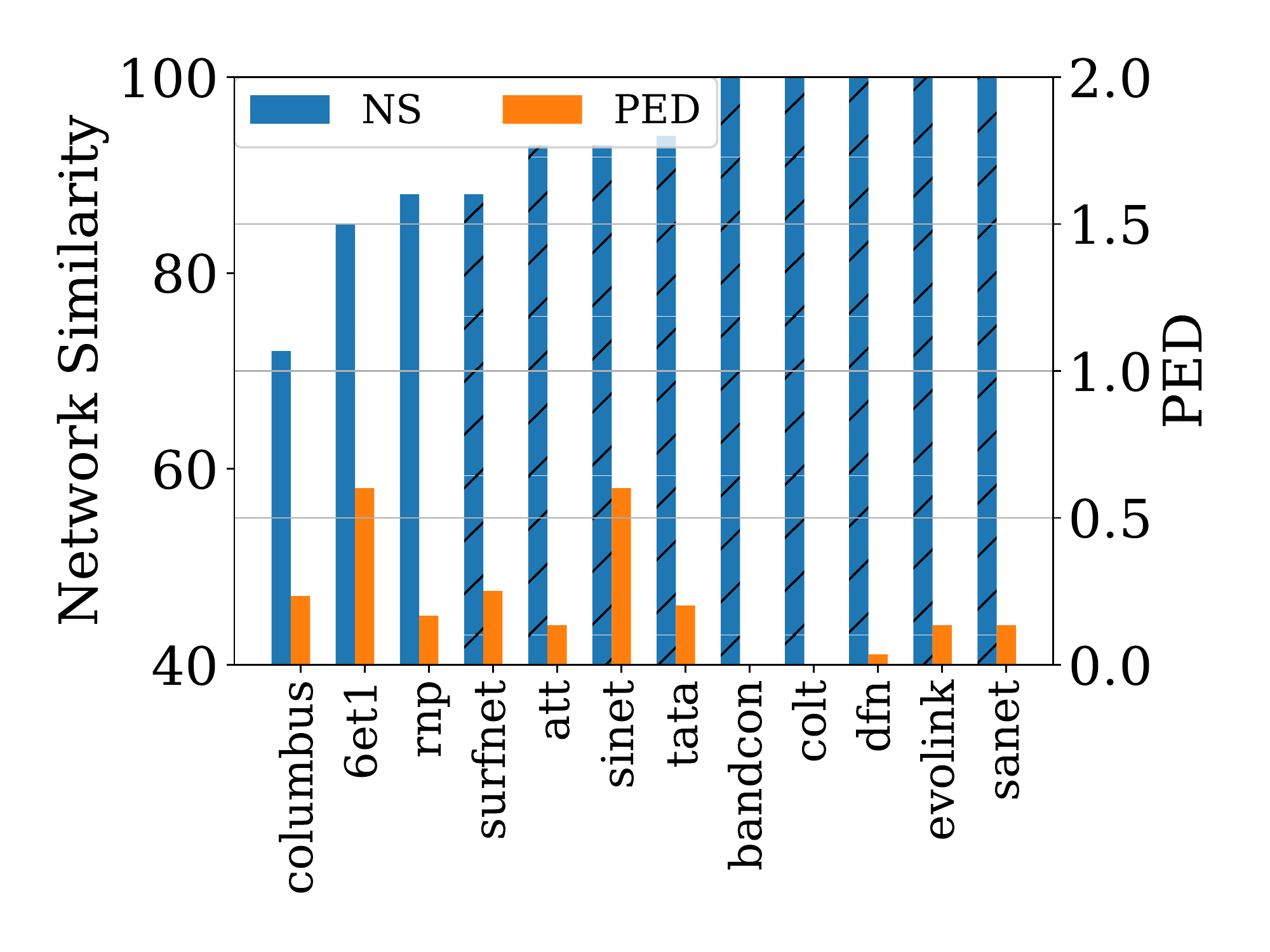}
\caption{Network Inference with both PSM's and DM's as  input}
\label{fig:dm_psm}
\end{figure}

\begin{figure}[!htb]
\centering
\includegraphics[width=0.9\linewidth, height=45mm]{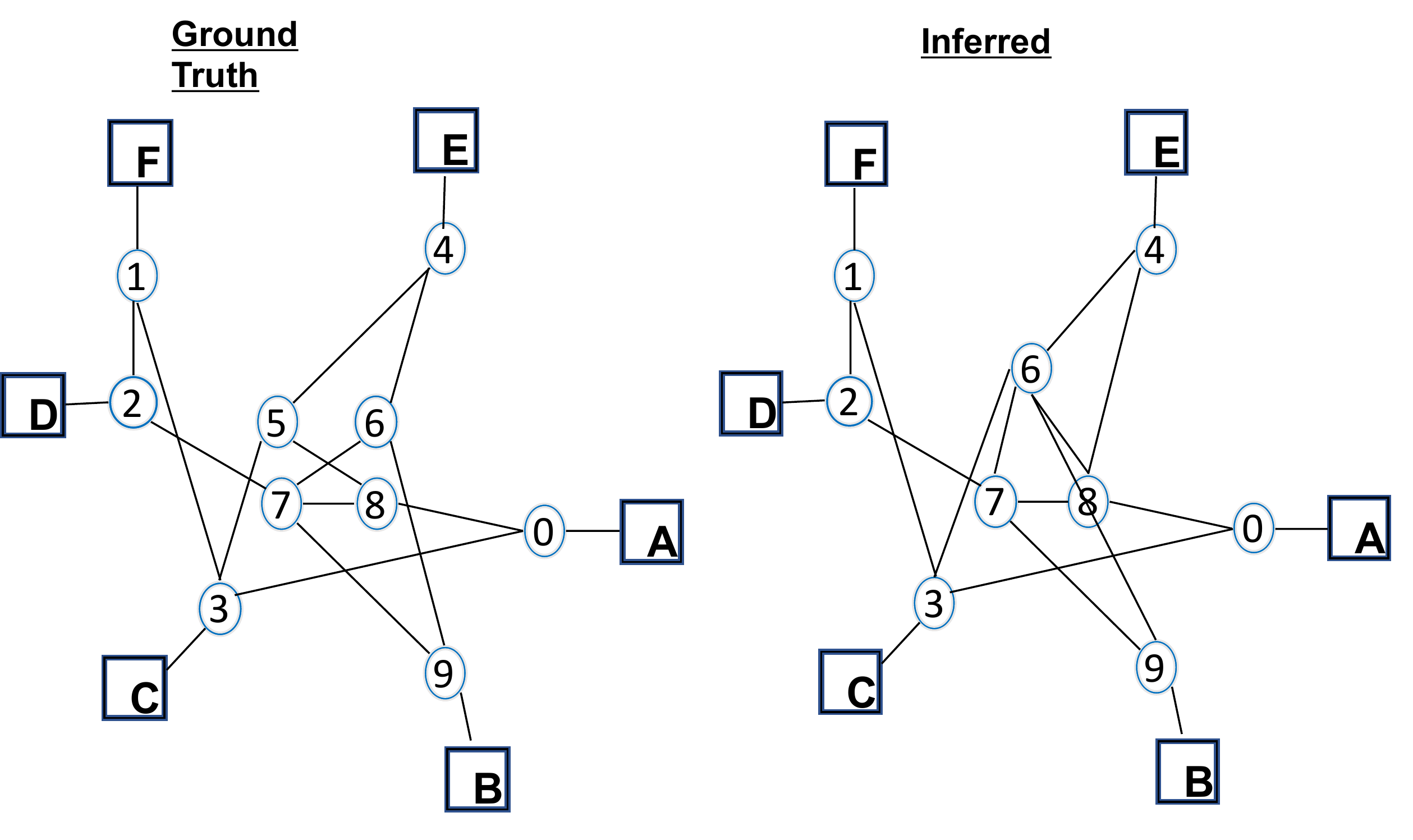}
\caption{COLUMBUS network inferred with an network similarity of 75\%.}
\label{fig:columbus}
\end{figure}

\begin{figure}[!htb]
\centering
\includegraphics[width=0.9\linewidth, height=45mm]{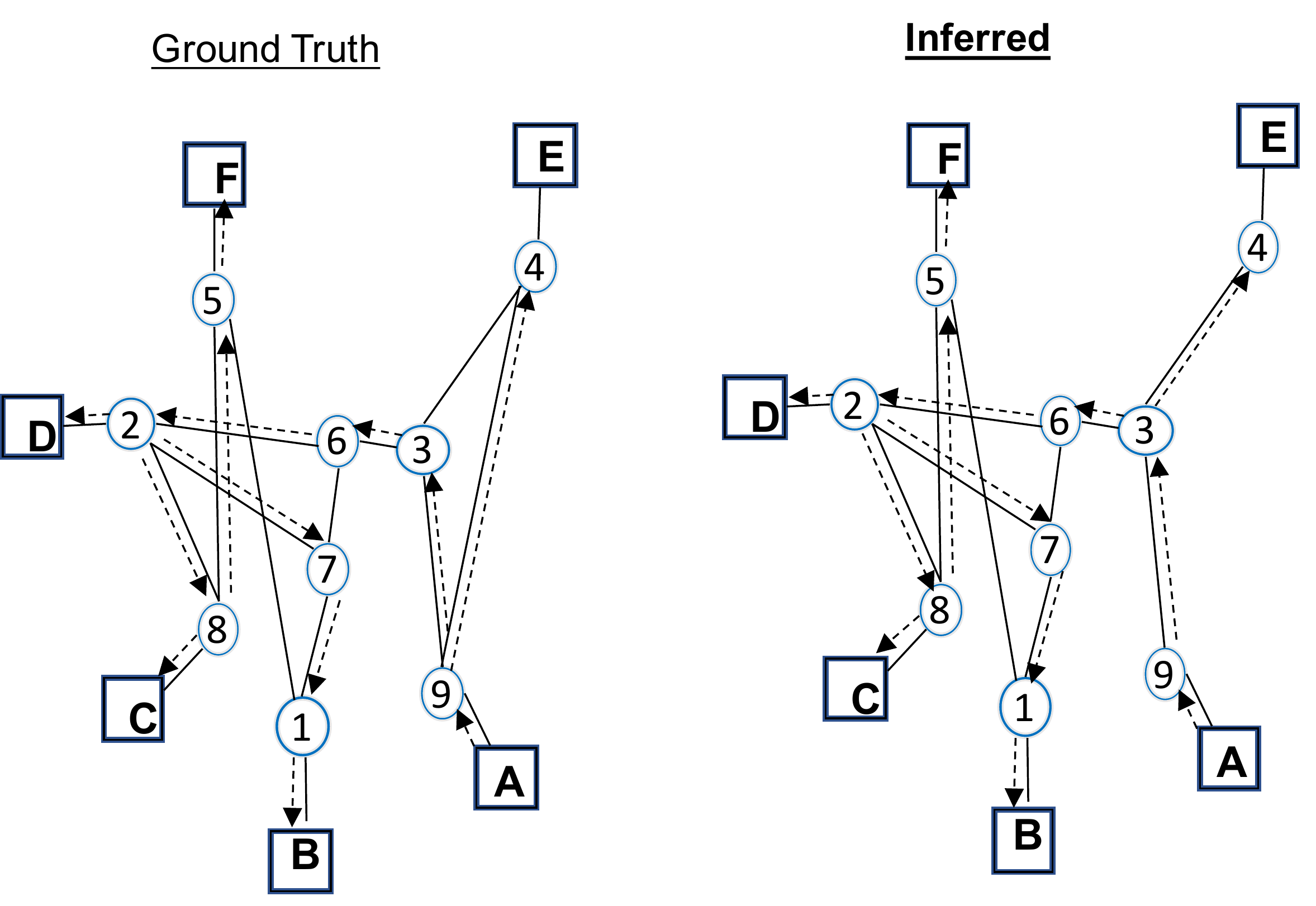}
\caption{TATA network inferred with an network similarity of over 94\%. The routing paths from source A to the hosts are also shown.}
\label{fig:tata}
\end{figure}

\subsection{Measurements from ground truth}
\label{sec:synthetic}

In the first set of experiments described in this section, we create the ground-truth communication network $N=(G,H,P)$ by picking a real-world network from topology-zoo (see Table~\ref{table:networks}). To simulate the situation where an enterprise has a set of hosts attached to the real-world network,  we choose a set of nodes randomly from the real-world network and attached a host to each of these nodes. To create the routes $P$, we find paths between the hosts   by computing shortest paths between every pair of hosts using Dijkstra's algorithm in a manner similar to OSPF \cite{Moy}.  The graph $G$ is simply the set of nodes and edges used in one or more of the shortest paths in $P$. Now that the ground-truth network $N$ is constructed, the measurement inputs to OCCAM are derived by computing the PSM and DM metrics from the ground truth $P$. Thus, this set of experiments model the situation where the measurement inputs to OCCAM have no errors, and only the ability of OCCAM to perform the optimization and inference is evaluated.
 
The ground truth and the inferred topology for AT\&T is shown in Figure~\ref{fig:att}. The inference is accurate with a network similarity (NS)  score of $93.75\%$, with the only error being an extra link (6,2) in the inferred graph not present in the original. The path edit distance (PED) was $0.6$, denoting the paths were also inferred accurately requiring only a small number of edits to make the inferred path identical to the corresponding path in the ground truth.

Figure \ref{fig:dm_psm} shows the overall performance of OCCAM across multiple networks. As can be seen, for a few networks we obtain a perfect inference, i.e., these networks received an NS score of 100\% and a PED of $0$. This means that both the inferred graph topology and the paths completed agreed with ground truth. Across the $12$ networks tested, we obtain an average NS score of 93\%. The average $PED$ score of $0.20$, which means that the average number of edits needed to make an inferred path identical to the same path in ground truth is $0.20$. Thus, OCCAM provides a highly accurate inference of the network, given accurate PSM and DM inputs from ground truth.

Beyond numerical measures, it is instructive to visualize the inferred networks themselves in relation to the ground truth in the cases where the inference was not perfect.   Figure \ref{fig:columbus} shows OCCAM's output for the COLUMBUS network that received one of the lower NS scores.  However, the inferred network and the ground truth have a very similar topological structure, except that internal nodes 5 and 6 in the ground truth are merged into one node (node 6) in the inferred graph. The merged internal node error is common since OCCAM attempts to find the ``simplest'' network that obeys the PSMs and DMs, resulting in OCCAM positing fewer internal nodes. Note that OCCAM does not infer a network with even fewer internal nodes, e.g., only one internal node instead of 5, 6, 7, and 8 in the ground truth, as such an inference will violate the DM and possibly some PSM constraints.  

As another example, Figure~\ref{fig:tata} shows the inferred network and ground truth for the TATA network.  OCCAM produces a nearly identical topology, except that the link $(9,4)$ is omitted in the inferred graph. The source tree rooted at $A$ is identical between the two graphs, though the path from $A$ to $E$ is longer in the inferred graph by one link. The reason for OCCAM's inference can be understood by the fact that its objective function in Equation~\ref{eq:objective} is a weighted sum of the number of links and shortest path distances. Since the shortest path between only one host-pair is impacted by not creating $(9,4)$ and since OCCAM was run with $\alpha = 0.2$ that favors link reduction over distance reduction, it chose not infer link $(9,4)$. Note that all the DM constraints are still met without  $(9.4)$, so the inferred network still meets all PSM and DM inputs.

\vspace*{-0.1in}
\subsubsection{Using DM inputs only} 
To observe the value of the DM inputs, we run OCCAM with only the DM constraints, without any PSM constraints. As shown in  Figure \ref{fig:dm}, the DM inputs by themselves provide an NS score of around 85\%. We also see that DM inputs are sometimes fully sufficient to obtain an accurate inference. For instance, on networks COLT, EVOLINK, SANET and SINET we obtain a $100\%$ score on the NS metric. In a few other networks, such as ATT and DFN, supplementing DMs with PSMs improves the inference significantly. For instance, NS score of the AT\&T network improved from 78\% to 94\%. Also, using just DMs, OCCAM infers the right number of internal nodes for 9 out of the 12 networks. Further, using the DMs alone provided an average PED of $0.43$ across the $12$ networks that we tested as compared to $0.20$ when both PSM and DM inputs are used.  In conclusion, DMs by themselves provide powerful constraints for network inference, though in several cases the PSMs improve inference quality.

\begin{figure}[!htb]
\centering
\includegraphics[width=0.95\linewidth, height=50mm]{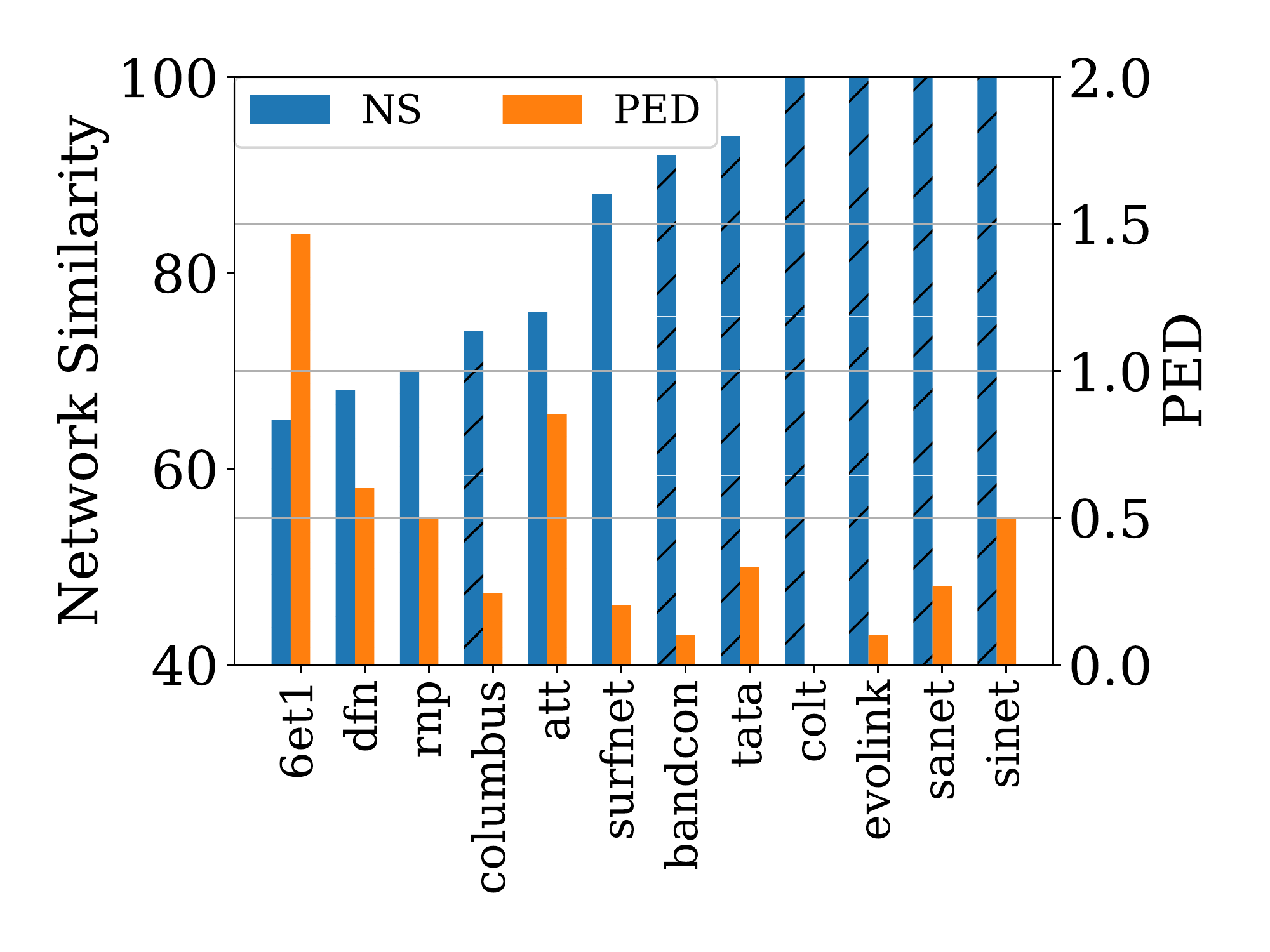}
\caption{Network Inference using only DMs}
\label{fig:dm}
\end{figure}

\begin{figure}[!htb]
\centering
\includegraphics[width=0.9\linewidth, height=38mm]{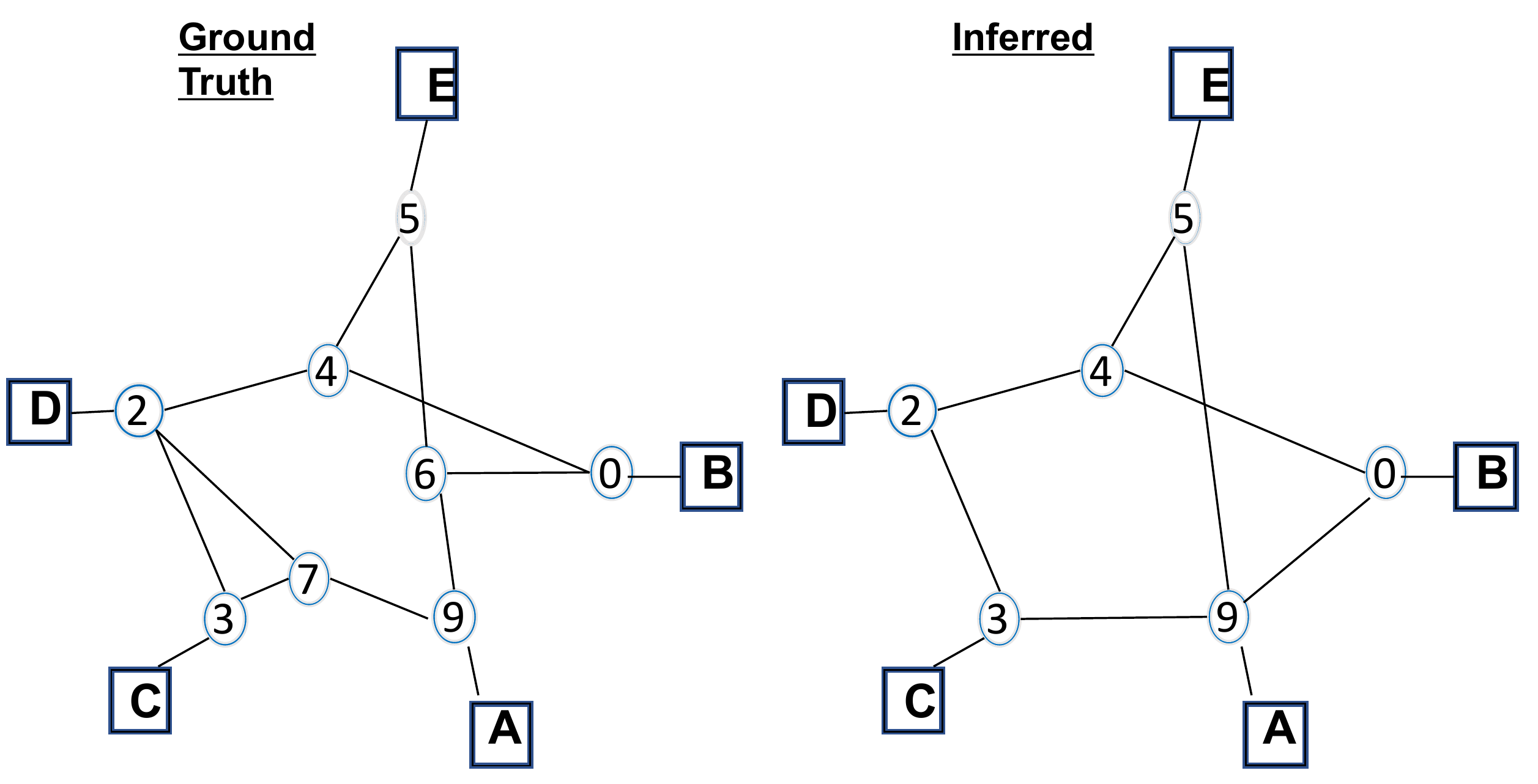}
\caption{AT\&T network inferred with only PSM inputs.}
\label{fig:attpsm}
\vspace*{-0.1in}
\end{figure}
\subsubsection{Using PSM inputs only} 
\label{sec:psmonly}
To observe the value of PSM metrics, we ran OCCAM with PSM inputs alone. Without the distance information from the DMs, we observe that OCCAM does not always guess the number of internal nodes correctly. Of the 12 networks, OCCAM with PSM alone produced the right number of internal nodes in only  5 cases, significantly less than OCCAM with DM inputs alone.  Figure~\ref{fig:attpsm} provides the OCCAM's inference for the AT\&T network with PSM inputs only. Unlike the case when both PSM and DM inputs are present (see Figure~\ref{fig:att}) where OCCAM deduced the right number of internal nodes, two pairs of nodes in the ground truth (nodes 7,3 and 6,9) are collapsed to a single node each in the inferred network in Figure~\ref{fig:attpsm}.  This example shows that DMs provide information for inferring internal nodes that cannot be inferred from PSMs alone. 

If we allow the collapse of internal nodes in the ground truth, OCCAM with PSM inputs does produce a high-quality inference of the network. To illustrate this point, we allowed up to two pairs of internal nodes to be collapsed in the ground truth graph before evaluating the NS and PED metrics. Figure~\ref{fig:psm} shows the NS score and PED values after allowing up to two pairs of nodes to be collapsed in the ground truth, where we choose the best pairs to collapse so as to optimize the NS and PED values. We observe that across 12 networks, we obtain a $NS$ score of 83\% and a PED of $0.44$. Thus, besides the error of collapsing internal nodes, OCCAM with PSMs alone can produce high-quality inferences. 

\begin{figure}[!htb]
\centering
\includegraphics[width=0.95\linewidth, height=50mm]{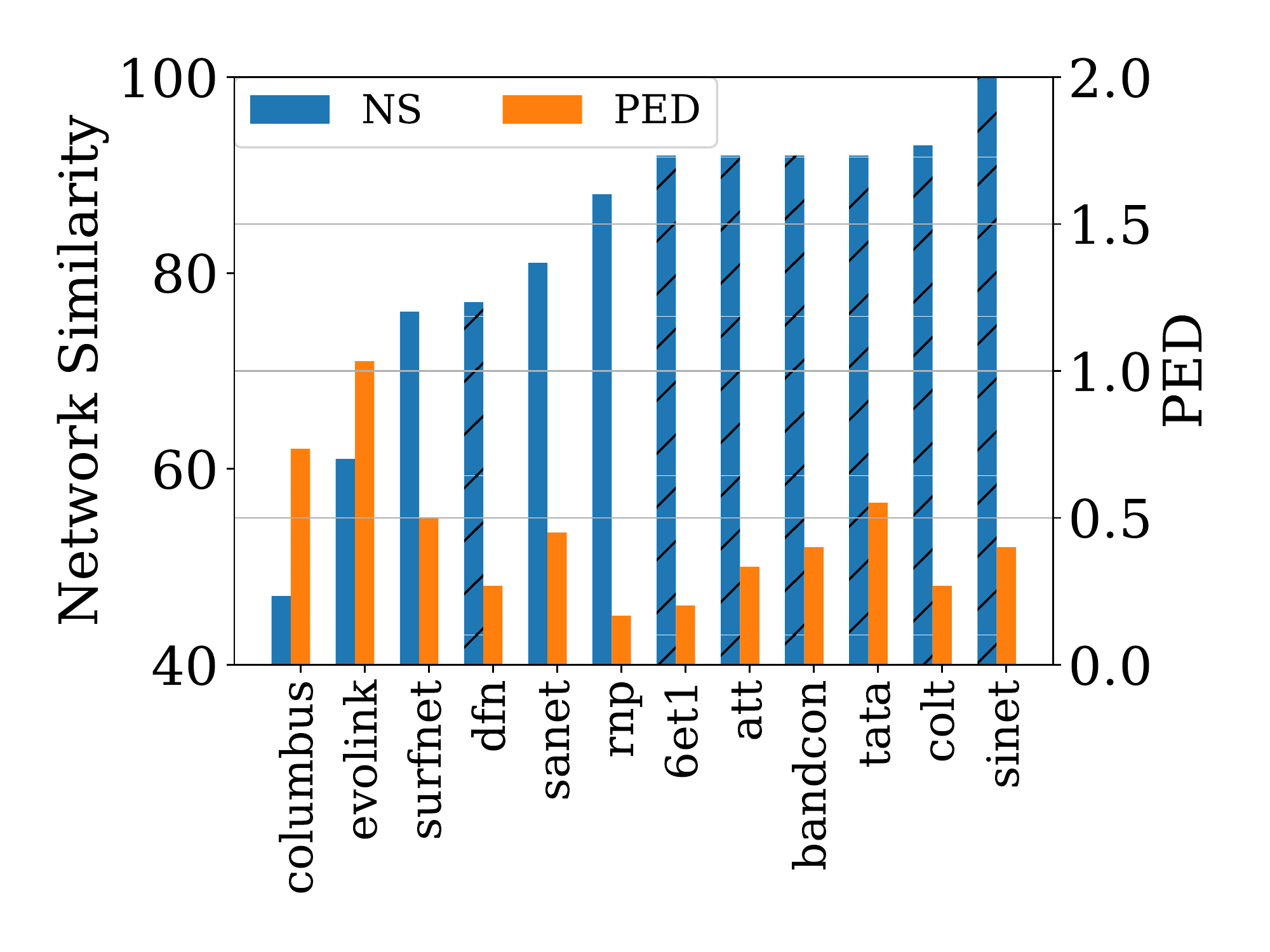}
\caption{Network Inference using only PSMs}
\vspace*{-0.2in}
\label{fig:psm}
\end{figure}

\subsection{Measurements with random errors}
Thus far, the PSMs and DMs derived from ground truth had no errors. However, when PSMs and DMs are derived from actual packet-level measurements in a real-life scenario, we expect some of them to be erroneous. Here, we study robustness of OCCAM's inference to erroneous PSM and DM inputs. We chose three network topologies (AT\&T, SANET, and BANDCON). As before, the measurement inputs to OCCAM are accurate PSM and DM metrics from ground truth. However, to introduce an error with probability $p\%$ in the relative PSM measurements, we chose each PSM constraint of the form shown in (\ref{eq:entanglement}) and flipped the LHS and RHS of that constraint with probability $p\%$. Likewise, we also flip the LHS and RHS of each DM constraint of the form shown in (\ref{eq:relative}) with probability $p\%$.  

Since these errors could introduce inconsistencies in the constraints leading to infeasibility, we use the PSM and DM constraints as ``soft'' constraints that are made a part of the objective function, as described in Section~\ref{subsec:opt}.  So, OCCAM finds the ``simplest'' network that satisfies as many (but not necessarily all) of the PSM and DM constraints as possible. Figures~\ref{fig:noisy_data} and \ref{fig:noisy_data_ped} show the quality of OCCAM's inference for the three networks with increasing error probability. As expected, the NS score decreases with increasing error. However, the NS scores are surprisingly good for the three networks, even in the presence of 20-30\% random input errors. Likewise, PED stays below $1$ for up to 40\% errors. Thus, the OCCAM approach allows accurate inputs to compensate for the incorrect ones to maintain a high inference quality.

\begin{figure}[!htb]
\centering
\includegraphics[width=0.9\linewidth, height=41mm]{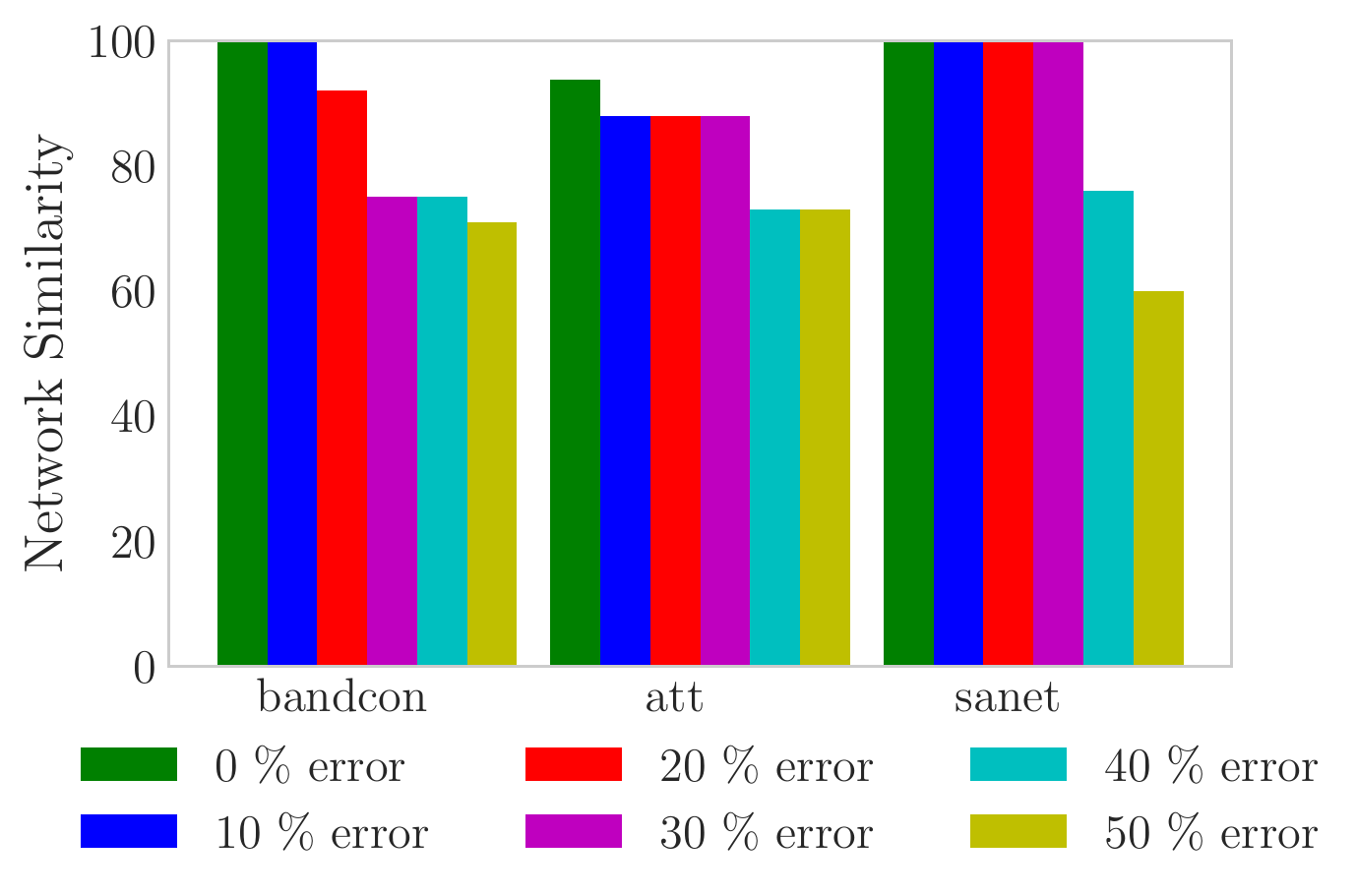}
\caption{NS scores for varying error probabilities}
\label{fig:noisy_data}
\vspace*{-0.15in}
\end{figure}

\begin{figure}[!htb]
\centering
\includegraphics[width=0.95\linewidth, height=41mm]{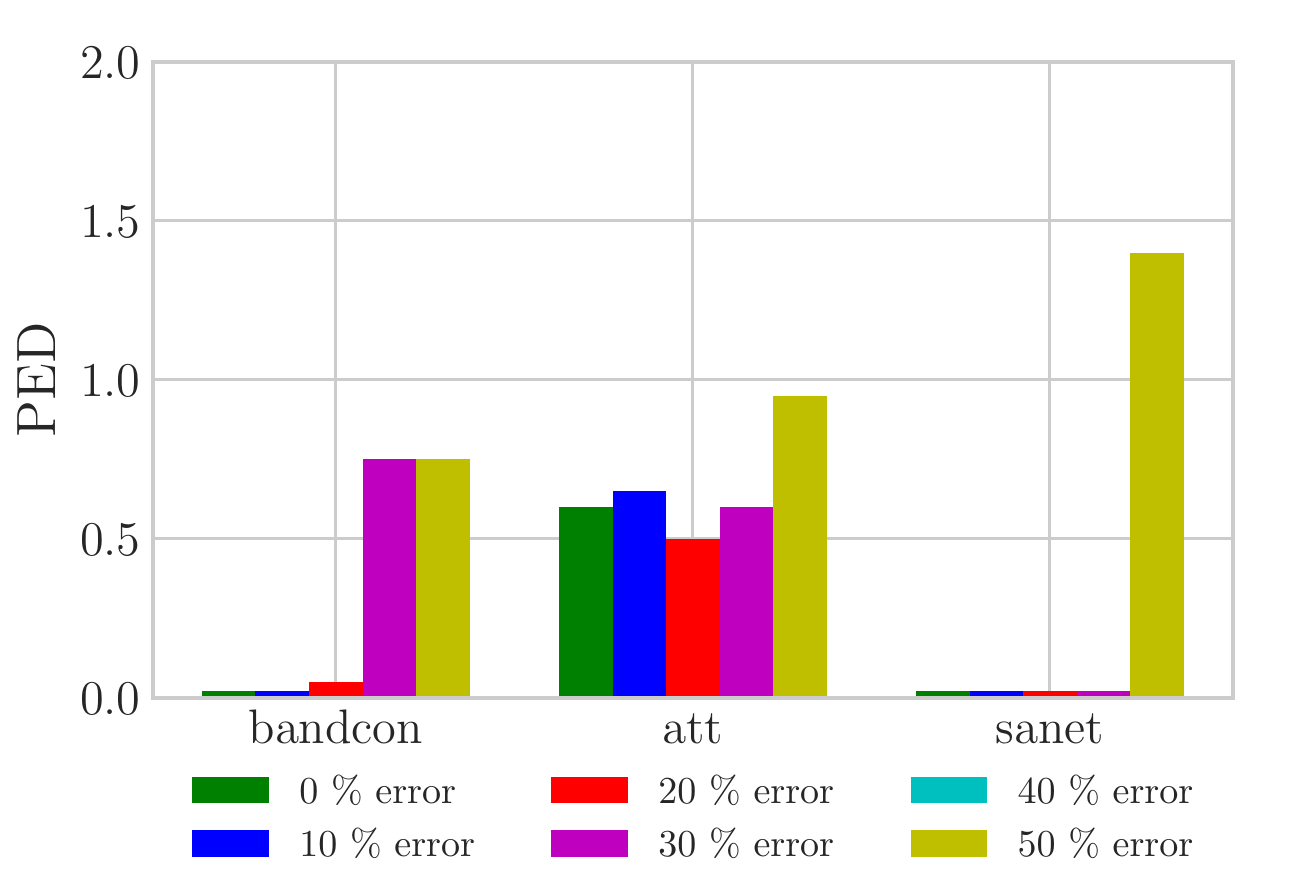}
\caption{PED for varying error probabilities}
\label{fig:noisy_data_ped}
\vspace*{-0.15in}
\end{figure}

\subsection{Measurements from a network emulator}
\label{sec:deter}
In this set of experiments, we evaluate OCCAM with its measurement inputs derived from a packet-level network emulator called DETER \cite{Mirkovic10the}. We use the TATA  and 6et1 topologies used in the previous sections and  configured them on DETER by specifying the nodes and links. DETER {\em automatically} sets up routing paths between the hosts that are specified, i.e., we do not specify the routes.  We obtain the PSMs and DMs from DETER  by sending packet probes. These experiments emulate the real-world situation where the network provider(s) set the routing paths and path metrics are derived from actual packet flows.

{\em PSM and DM metrics.} We use a standard delay covariance technique to measure PSMs.For every source $S$ and receivers $T1$, $T2$ and $T3$, we compare $PSM(S,T_1,T_2)$ and $PSM(S,T_2,T_3)$ by performing the following experiment. A train of three back-to-back packets, $p_{i1}$, $p_{i2}$ and $p_{i3}$ destined to receivers $T_1$, $T_2$ and $T_3$ respectively are sent from source $S$. At each receiver $T_j$, a record of the delays experienced by packets $p_{ij}$ is maintained. We calculate the delay covariance $C(S; T_i, T_j)$ between packets received at receivers $T_i$ and $T_j$. If it is observed that $C(S; T_1, T_2) > C(S; T_2, T_3)$, then a PSM constraint  $PSM(S,T_1,T_2) > PSM(S, T_2, T_3)$ is added to the optimization. The DMs are obtained from the $TTL$ field  \cite{Postel1981} of the IP Header in a standard way as described in Section~\ref{sec:pathmetrics}.

The results of the DETER experiments for TATA and 6et1 are shown in Figures~\ref{fig:tata_deter} and \ref{fig:6et1_deter} respectively. In the case of TATA, OCCAM inferred extra links (8,2) and (7,0), but rest of the inference was accurate, receiving an NS score of 89.4\% and a PED of 0.7. In the case of 6et1, nodes 8 an 6 in the ground truth were merged in the inferred network, but rest of the inference was accurate. The NS score was 88.8\% and PED was 0.66. Notably, for the 6et1  (resp., TATA) network, 18\% (resp., 26 \%) of the PSM constraints were incorrect due to measurement error. However, despite the errors in its input. OCCAM produced high quality inferences.
\begin{figure}[!htb]
\centering
\includegraphics[width=0.9\linewidth, height=46mm]{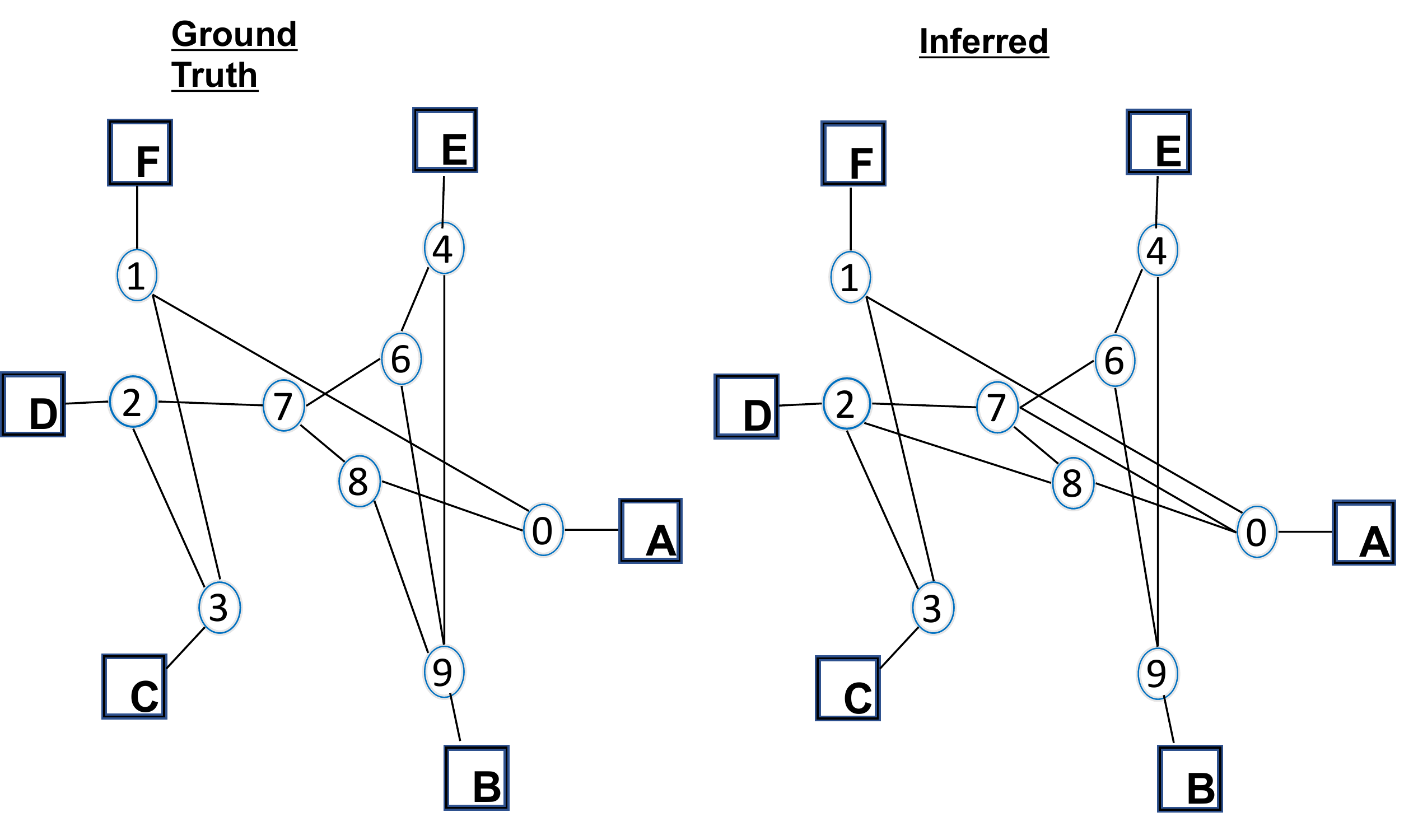}
\caption{TATA network on DETER achieved an NS score of 89.4\% and a PED of 0.7}
\label{fig:tata_deter}
\vspace*{-0.1in}
\end{figure}

\begin{figure}[!htb]
\centering
\includegraphics[width=0.95\linewidth, height=46mm]{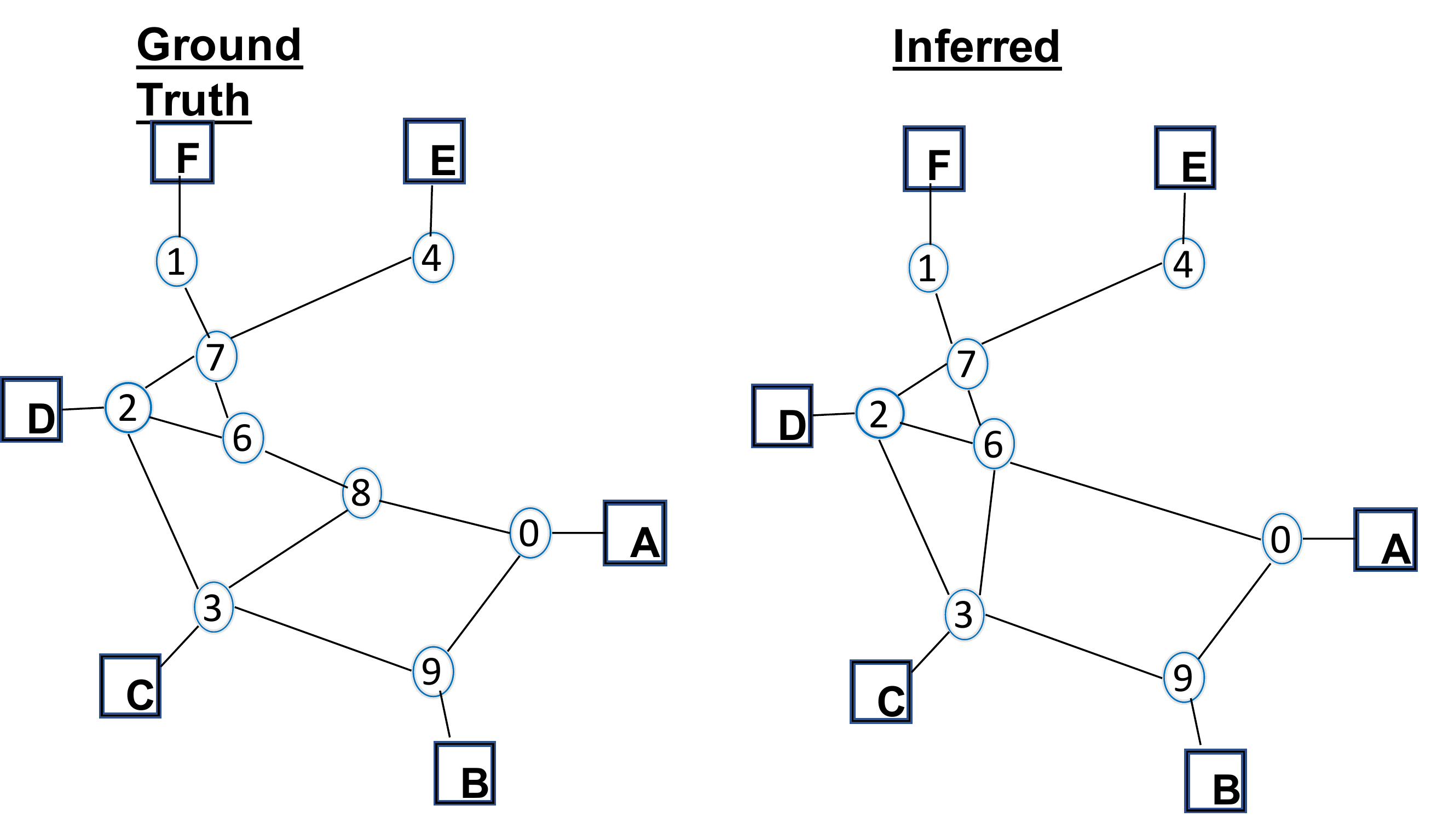}
\caption{6et1 network on DETER achieved an NS score of 88.8\% and PED of 0.66.}
\label{fig:6et1_deter}
\vspace*{-0.05in}
\end{figure}

%% file: tree.tex
\section{The Tree Stitching Problem}
A variant of the network inference problem that we call the ``tree-stitching problem'' builds directly on classic network tomography results that infer source trees from PSM metrics. In this variant,  we use well-known methods \cite{Ratnasamy1999, Duffield2004, Duffield2002} to create source trees rooted at each source host $S$ and provide these trees as measurement inputs to OCCAM, in lieu of the PSM metrics.  OCCAM ``stitches'' together these source trees to infer a network as described below.

%a network $N$ is said to be \textit{consistent} with the given source trees, if for each source $S \in H$, there exists a one-to-one mapping $\delta: {\mathcal B}^S \rightarrow V$; and %each segment $s \in {\mathcal T}^S$ is a path in $G$.

\begin{figure}[!htb]
\centering
\includegraphics[width=0.50\linewidth, height=46mm]{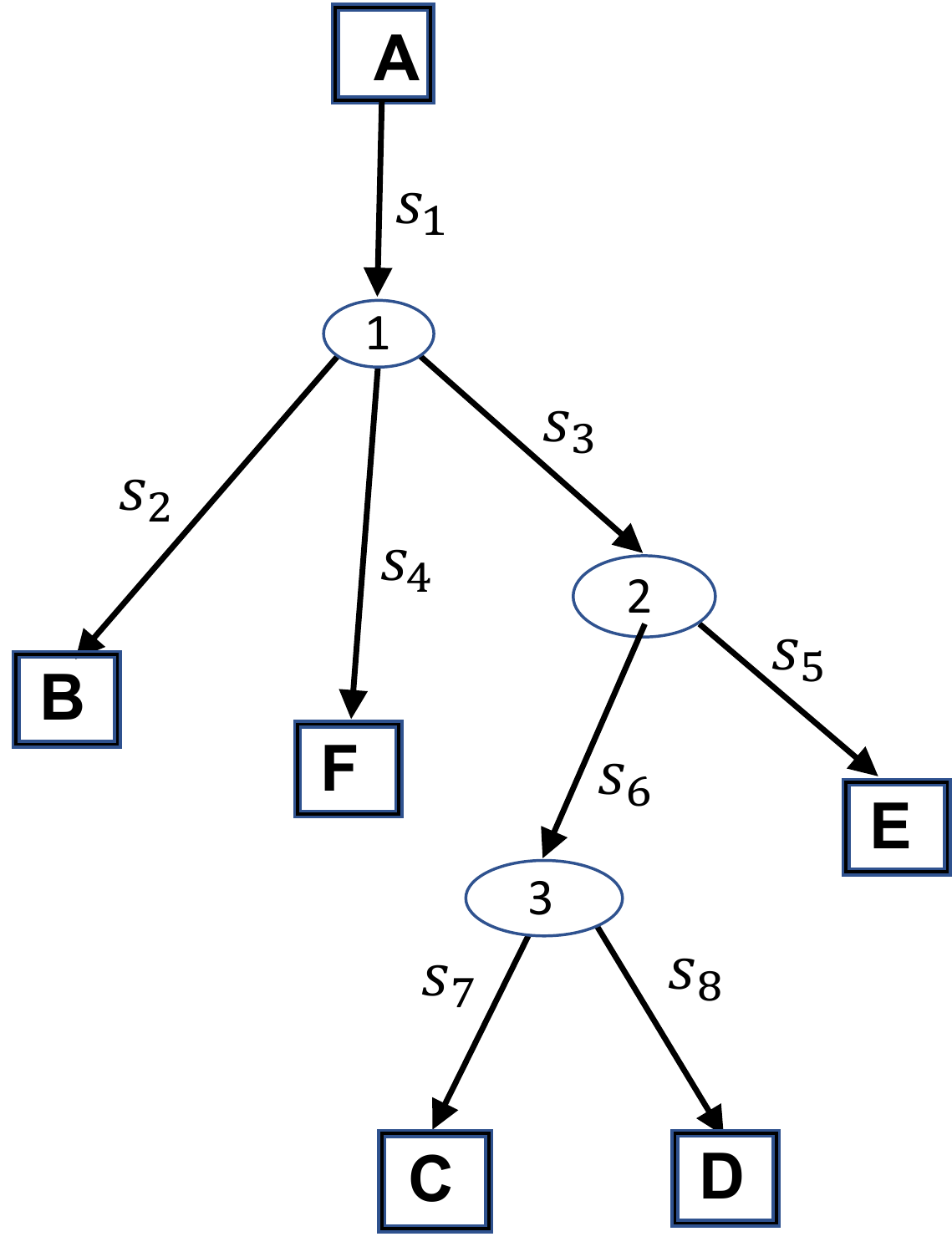}
\caption{A source tree can be viewed as a set of segments and branch points.}
\label{fig:srctree}
\end{figure}
\vspace*{-0.2in}

\subsection{Optimization Step}
The main challenge is to add constraints that ensure that the inferred network $N' = (G', H, P')$ is consistent with all the source trees provided as measurement inputs, i.e., the logical source tree formed by the paths in $P'$ from each source $S$ is isomorphic to the given source tree rooted at $S$.

A source tree can represented by a set of segments and branch points (see Figure \ref{fig:srctree} for an example).  Each segment $s \in {\mathcal T}^S$ represents one or more links in the underlying graph. Segment $s$ terminates at a unique branch point $b_s \in {\mathcal B}^S$ and branches into a set of outgoing segments ${\mathcal O}(s)$. For instance, in Figure \ref{fig:srctree}, $s_{1}$ is the incoming segment at branch point $1$, and ${\mathcal  O}(s_1)= s_{2}, s_{3}, s_{4}$ are the outgoing segments.   A source tree can be completely characterized by adding constraints to capture the segments ending at each branch points $b_s \in {\mathcal B}^S$.  For every segment $s \in {\mathcal T}^S$ terminating at branch point $b_s$, we add the following constraint:

\begin{equation}\label{eq:logical-trees}
\begin{split}
p_{i,j}^{s} \leq \sum_{k \in V} p_{j,k}^{s} +  \frac{1}{|{\mathcal O}(b_s)|} \sum_{s' \in {\mathcal O}(b_s)}\sum_{k \in V} p_{j,k}^{s'}
				 = b_{j}^{s} \\ \quad \forall i \in V, \quad \forall j \in  V \setminus H,
\end{split}
\end{equation}
where $p_{i,j}^{s}$ is an indicator variable indicating whether link $(i,j)$ is present in segment $s$. The above constraint (\ref{eq:logical-trees}) ensures, if there exists an incoming link at $j$ which is part of $s$, then there must exist an outgoing link $(j,k)$ which either belongs to same segment $s$ (first term above) or to one of the outgoing segments $s' \in O(b_s)$(second term above).
For every segment $s$ in ${\mathcal T}^S$ that ends at a host $T \in H$, we add the following constraint.
\begin{equation}\label{eq:logical-trees-end}
\begin{split}
p_{i,j}^{s} \leq \sum_{k \in V} p_{j,k}^{s} \quad \forall i \in V, \quad  \forall j \in  V \setminus T,
\end{split}
\end{equation}

If segment $s$ terminates at a host $T \in H$, the above constraint ensures there exists an outgoing link $(j,k)$ which belongs to same segment $s$,  if there exists an incoming link at $j$ which is part of $s$, unless $j = T$.
\begin{comment}
\begin{equation}\label{eq:branch-point}
\sum_{k \in V} p_{j,k}^{s} +  \frac{1}{|{\mathcal O}(b_s)|} \sum_{s' \in \textcolor{blue}{{\mathcal O}(b_s)}} \sum_{k \in V} p_{j,k}^{s'} \quad \forall j \in  V \setminus H
\end{equation}
\end{comment}
At each node $j \in V$, we ensure that there could be at most $1$ outgoing link that belongs to a segment $s$.
\begin{equation}\label{eq:one-out-link}
\sum_{k \in V} p_{j,k}^{s} \leq 1 \quad \forall j \in V
\end{equation}

\textbf{Boundary Conditions.}  For a segment $s \in {\mathcal T}^S$ originating at root $S$ we add the following constraint,
\begin{equation}\label{eq:src-boundary}
\begin{split}
\sum_k p_{S,k}^s = 1,
\end{split}
\end{equation}
which ensures that the segment will always contain an outgoing link $(S,k)$ for some $k \in V$.
For a segment $s \in {\mathcal T}^S$ terminating at host $T \in H$, we add the following constraints,
\begin{equation}\label{eq:dst-inc-boundary}
\begin{split}
\sum_j p_{j,T}^s = 1 \quad \forall T \in H \setminus S,
\end{split}
\end{equation}
\begin{equation}\label{eq:dst-boundary}
\begin{split}
\sum_i p_{T,i}^s = 0 \quad \forall T \in H \setminus S,
\end{split}
\end{equation}
where, constraints in (\ref{eq:dst-inc-boundary}) ensures there always exists an incoming link at $T$ that belongs to segment $s$ and constraints in (\ref{eq:dst-boundary}) ensures that segment terminating at the terminal $T \in H$ will not contain an outgoing link $(T,k)$ for any $k \in V$.

\begin{algorithm}\label{al:network-construct2}
\SetAlgoLined
\caption{\textit{GRAPH-CONSTRUCT-II}}
\textbf{Input :} Set of hosts H and solution variables \;
\textbf{Output :} Network $N'=(G',H,P')$. \;
\BlankLine
Initialize: $V' \leftarrow \phi$, $E' \leftarrow \phi$, $P' \leftarrow \phi$ \;
\BlankLine
\ForEach{$S\in H$} {%
	\BlankLine
	$P^S \leftarrow \phi$ \;
	\ForEach{s in ${\mathcal T}^S$}{
		$L = \{(i,j) : p_{i,j}^{s} = 1$ \} \;
		$P^S \leftarrow P^S \cup L$ \;
	}	
	\ForEach{$(i, j) \in P^S$}{ 
		$E' \leftarrow E' \cup (i,j)$ \;
		$V' \leftarrow V' \cup \{i\} \cup \{j\}$ \; 
	} 
	\ForEach{$T \in H$}{
		$\pi(S,T) \leftarrow$ shortest path from $S$ to $T$ in $P^S$ \;
		$P' \leftarrow P' \cup \pi(S,T)$ \;	
		\BlankLine
	}
}
$G' \leftarrow (V', E')$ \;
\Return $N'=(G',H,P')$ \;
\end{algorithm}

%\vspace*{-0.1in}

To maintain consistency, we need to make sure that if link $(i,j)$ belongs to any of the segments in ${\mathcal T}^S$, it should also belong to $P^S$. Further, we want to ensure that each link $(i,j)$ is part of at most one segment. To enforce these properties, we add the following constraint 
\begin{equation}\label{eq:src-consistency}
\begin{split}
\sum_{s \in {\mathcal T}^S} p_{i,j}^{s} = s_{i,j}^S \quad \forall (i,j) \in E,
\end{split}
\end{equation}
where $s_{i,j}^S$ is a binary variable indicating if link $(i,j)$ belongs to $P^S$.  The \textit{LHS} of the above equation counts the number of segments in ${\mathcal T}^S$ in which link $(i,j)$ is present. The above equation forces the number of such segments to be at most $1$ (as $RHS$ is a binary variable), and if such a segment exists, $s_{i,j}^S$ is set to $1$.

For the tree-stitching version of the problem, OCCAM's optimization step is modified to use (\ref{eq:objective}) as the objective function and (\ref{eq:relative}),(\ref{eq:src-tree})-(\ref{eq:boundary-absence-dst}), (\ref{eq:logical-trees})-(\ref{eq:src-consistency}) as the constraints.

\subsection{Inference step}

Algorithm GRAPH-CONSTRUCT-II infers a network $N' = (G', H, P')$ using the values of the $p_{i,j}^s$ variables in the solution of the optimization step. For every host $S \in H$, lines 6-9 in the algorithm finds the set of links $P^S$ such that each link $(i,j) \in P^S$ belongs to some segment $s \in {\mathcal T}^S$ (i.e., $p_{i,j}^s = 1$ ). Now, for every host $T \in H$, $\pi(S,T)$ is computed as the shortest path from $S$ to $T$ in $P^S$.  Each link $(i,j) \in P^S$ and the corresponding nodes $i$ and $j$ are added to the set of links $E'$ and the set of nodes $V'$ respectively. Thus at the end of the for loop in line 19, the algorithm infers routing paths $\pi(S,T) \in P'$ between every pair of hosts $(S,T) \in H \times H$, the set of links $E'$ and the set of nodes $V'$ in $G'$.

\subsection{Correctness of Tree Stitching}
We prove that the modified version of OCCAM presented in this section performs tree stitching correctly.
\begin{theorem} \label{th:tree} Given source trees and DMs as inputs, OCCAM infers a network $N'=(G',H,P')$; such that the routing paths P' satisfy the following properties
\begin{enumerate}
\item Each routing path $\pi(S,T) \in P'$ is an acyclic path from host $S$ to $T$; and
\item G' and P' is \textit{consistent} with the given source trees and DM measurement inputs.
\end{enumerate} 
\end{theorem}
\begin{pfsketch}
 Consider a host $S \in H$. Let $s_1 \in {\mathcal T}^S$ be a segment beginning at host $S$. Constraints in (\ref{eq:src-boundary}) ensure there exists a link $(S, i)$ that belongs to segment $s_1$, i.e., $p_{S,i}^{s_1}$ equals 1. Now constraints in (\ref{eq:logical-trees}) ensures that there exists a path $\pi(S, i_n)$ such that each link in $\pi(S, i_n)$ belongs to segment $s_1$. Thus, path $\pi(S, i_n)$ can be mapped to segment $s_1$ and node $i_n$ can be mapped to the branch-point $b_{s_1}$. The constraint (\ref{eq:logical-trees}) also ensures that there exist outgoing links at node $i_n$, such that each such outgoing link belongs to a segment $s' \in {\mathcal O}(s)$. Now for each segment $s' \in {\mathcal O}(s)$, the constraint ensures that there exists a path $\pi(i_n, i_{n_{s'}})$ such that each link in $\pi(i_n, i_{n_{s'}})$ belongs to segment $s'$. Thus, path $\pi(i_n, i_{n_{s'}})$ can be mapped to segment $s'$ and branch-point $b_{s'}$ can be mapped to node $i_{n_{s'}}$. Thus, each segment $s \in {\mathcal T}^S$ can be mapped to a path in $G'$ and each branch-point $b_s$ can be mapped to a node $v \in V'$. This shows that the source trees derived from $P'$ is isomorphic to the source trees provided as inputs. Constraints in (\ref{eq:relative}) ensure that DMs are satisfied. We provide the complete proof in Appendix \ref{app:correctness2}.
\end{pfsketch}

\begin{comment}
\textcolor{red}{
\subsection{Reconstruct-ability}
\begin{theorem} If the ground truth network $N=(G,H,P)$ is a tree, then given source trees and DMs as inputs, OCCAM uniquely reconstructs the network $N=(G,H,P)$ using just the DMs.
\end{theorem}
\begin{pfsketch}
\end{pfsketch}
}
\end{comment}

\subsection{Empirical results}
We ran OCCAM with the source trees and DM inputs.  Figure~\ref{fig:tree1} shows the results where OCCAM achieves an average NS score of 92.9\% and PED of 0.22, across the 12 tested networks. Thus, comparing Figures~\ref{fig:dm_psm} and \ref{fig:tree1}, using source trees and DMs as inputs yields similar results to using PSM metrics and DMs, i.e., source trees can take the place of PSM measurements and vice-versa. 

We also ran OCCAM with the source trees only as input, without DM inputs. Our results were similar to the case described in Section~\ref{sec:psmonly} where only PSM inputs were used.  As in Section~\ref{sec:psmonly}, the NS score and PED score was computed after allowing at most two node pairs to be contracted in the ground truth network. The results shown in Figure~\ref{fig:tree2} shows the results where OCCAM achieves an average NS score of 81.5\% and PED of 0.48 across the 12 tested networks. These results indicate that just source trees as inputs yields similar results as using just the PSM inputs. 

\begin{figure}[!htb]
%\vspace*{-0.2in}
\centering
\captionsetup{justification=centering}
\includegraphics[width=0.95\linewidth, height=50mm]{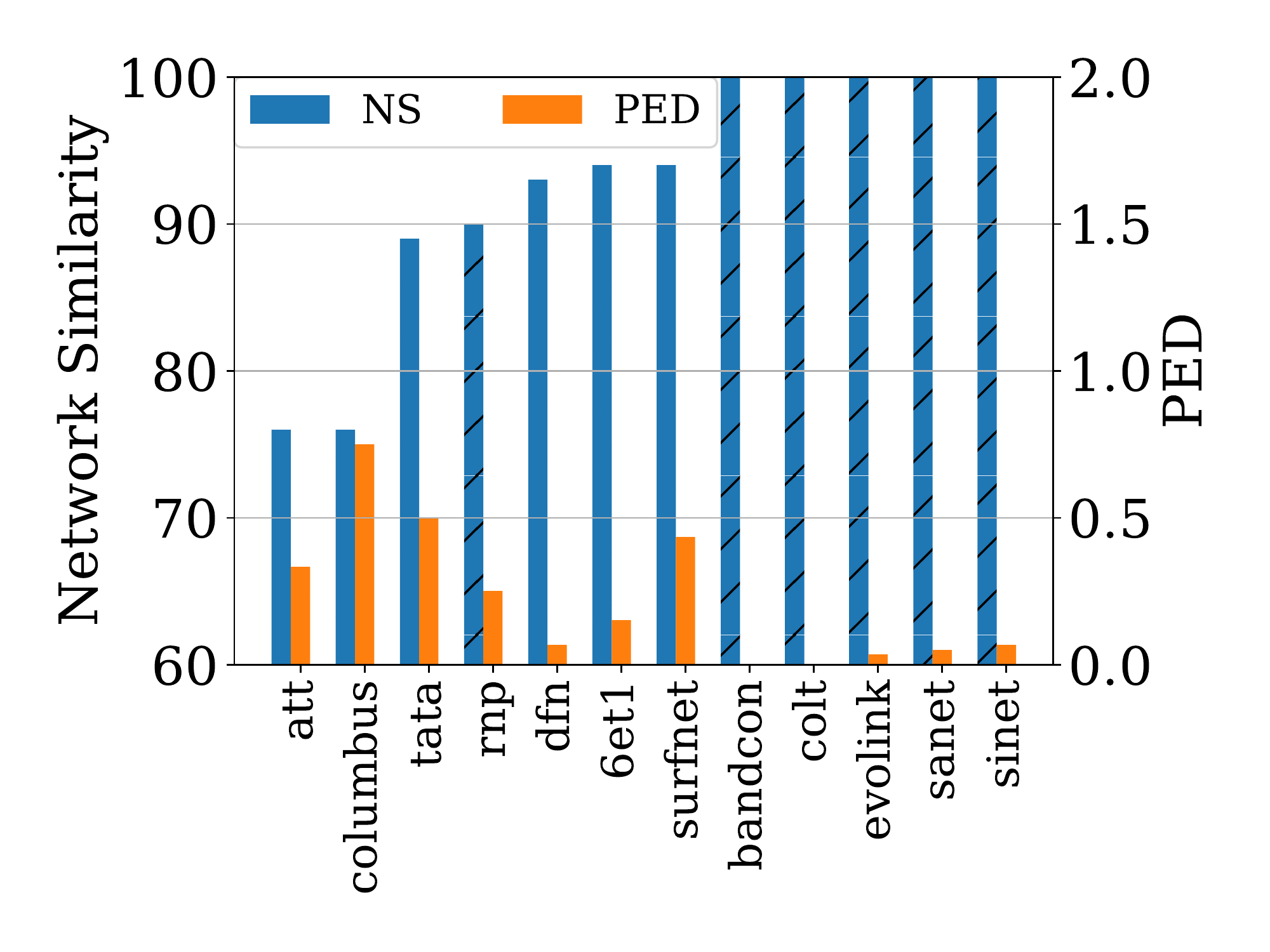}
\caption{Tree stitching with DM}
\vspace*{-0.1in}
\label{fig:tree1}
\end{figure}

\begin{figure}[!htb]
%\vspace*{-0.2in}
\centering
\includegraphics[width=0.95\linewidth, height=50mm]{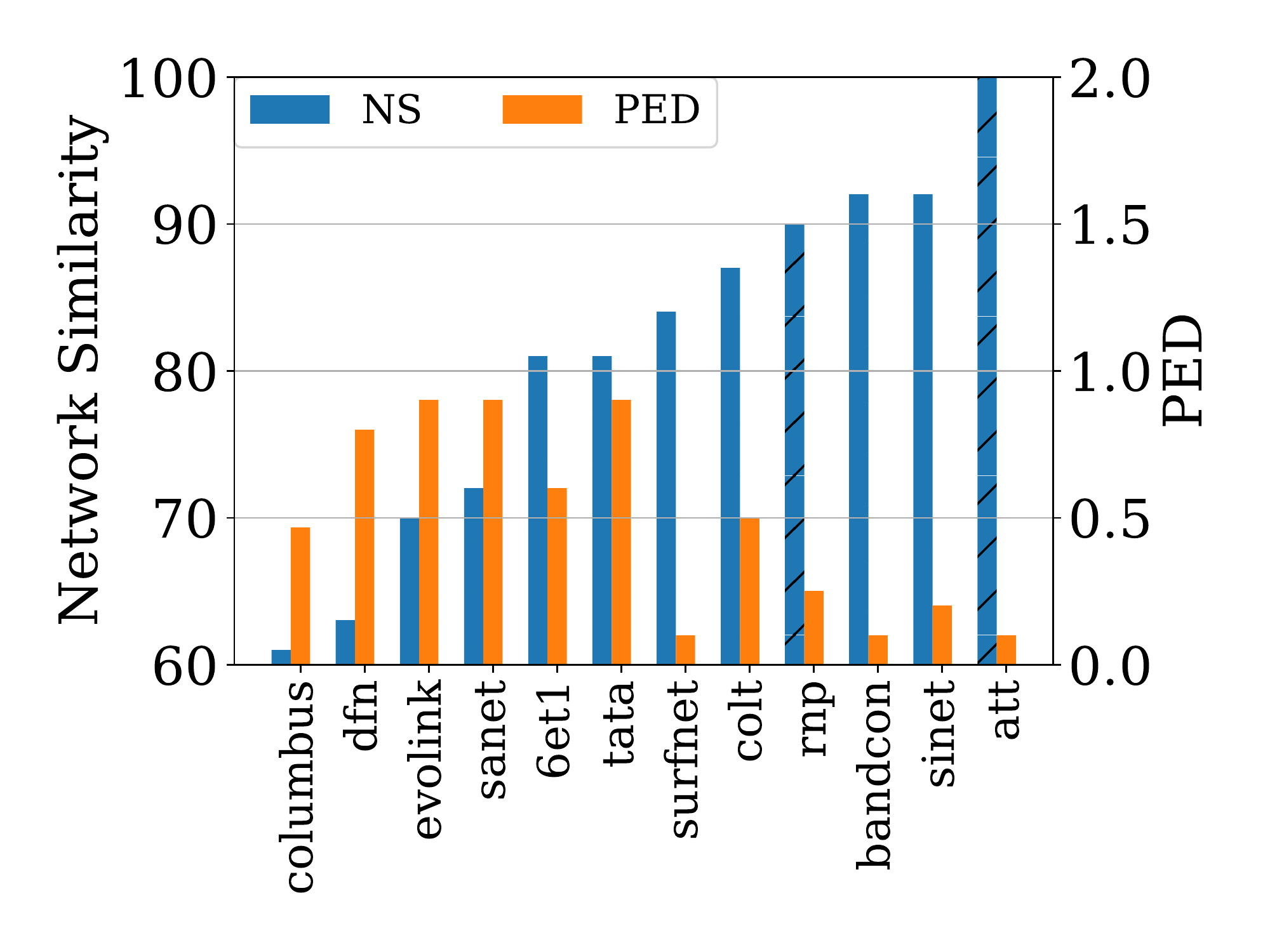}
\caption{Tree stitching without DM}
\vspace*{-0.1in}
\label{fig:tree2}
\end{figure}

%% file: conclusion.tex
\section{Conclusion}
Our work is the first to demonstrate the feasibility of inferring the complete network topology and routing paths using path sharing  and path distance information measured at the hosts.  However, many questions remain open. An interesting question is how much path sharing information is needed to infer the complete network. Our preliminary work suggests that highly accurate inference is possible even with partial and/or incorrect path sharing information. Another natural extension of our work is whether the topology and path inference can be extended to infer the link capacities in the network.

\begin{comment}
\section{Acknowledgements}
This research was supported in part by  DARPA under Contract No. N66001-15-C-4045 and NSF grant CNS-1413998.
\end{comment}

%% file: appendix1.tex
\section{Correctness of OCCAM}
\label{app:correctness}
We prove the correctness of OCCAM by formally showing that it provides a solution that satisfies all PSM and DM observations.

\begin{lemma} \label{l:cycle-free} For a host $S \in H$, if ${\mathcal L}^S = \{(i,j):s_{i,j}^{S}=1\}$, then ${\mathcal L}^S$ does not contain a cycle. 
\end{lemma}
\begin{proof} (Proof by Contradiction) Let ${\mathcal L} \subset {\mathcal L}^S$ consist of $n$ links $(i_1, i_2)$,$(i_2, i_3)$,$....,(i_{n-1}, i_n),  (i_n, i_1)$ that form a cycle.  Constraint in (\ref{eq:distance}) ensures,
\begin{align*}
m_{i_2}^{S} = \sum_{j \in V} s_{j, i_2}^S (m_{j}^{S} + 1)
\end{align*}
Now constraints in (\ref{eq:src-tree}) ensure there exists at most 1 incoming link at node $i_2$ (i.e, $\sum_{j \in V} s_{j, i_2}^S$ equals 1) ; and as link $(i_1, i_2) \in {\mathcal L}$,  $s_{{i_1},{i_2}}^{S} = 1$. The above equation thus reduces to, 
\begin{align*}
m_{i_2}^{S} = m_{i_1}^{S} + 1,
\end{align*}
Similarly, 
\begin{align*}
m_{i_3}^S = m_{i_2}^S + 1 = m_{i_1}^S + 2, 
\end{align*}  
\begin{align*}
m_{i_n}^S = m_{i_{n-1}} + 1 = m_{i_1}^{S} + n - 1,
\end{align*}
As link $(i_n, i_1) \in {\mathcal L}^S$,  
\begin{align*}
m_{i_1}^S = m_{i_n}^S + 1 = m_{i_1}^S + n, 
\end{align*}
\begin{align*}
\implies n = 0 
\end{align*}
This contradicts our assumption of existence of a cycle in ${\mathcal L}^{S}$ and thus completes the proof.
\end{proof}

\begin{lemma}\label{l:v-assign-mst} For a pair of hosts $(S,T) \in H \times H$,  $m_{T}^{S}$ is assigned a value equal to the length of the routing path $\pi(S,T) \in P'$.
\end{lemma}
\begin{proof} Let $\pi(S,T) = \{ i_0, i_1, i_2, ..... , i_{n-1}, i_n \}$ be a path of length $n$. We will first prove that for any node $i_k $ on path $\pi(S,T)$, $m_{i_k}^{S} = k$. We prove by induction. 

By definition $m_{S}^S = m_{i_0}^S = 0$. Constraints in (\ref{eq:distance}) ensures,
\begin{align*}
m_{i_1}^{S} = \sum_{j \in V} s_{j, i_i}^S (m_{j}^{S} + 1)
\end{align*}
Now constraints in (\ref{eq:src-tree}) ensure there exists at most 1 incoming link at node $i_1$ (i.e., $\sum_{j \in V} s_{j, i_i}^S$ equals 1); and as link $(i_0, i_1) \in \pi(S,T)$, $s_{{i_0},{i_1}}^{S} = 1$. The above equation reduces to,
\begin{align*}
m_{i_1}^{S} = m_{i_0}^{S} + 1 = 1
\end{align*}
Thus $m_{i_1}^{S} = 1$. Now lets assume $m_{i_k}^{S} = k$ for some $k \in [1, n)$.  With a similar argument as above,
\begin{align*}
m_{i_{k+1}}^{S} = m_{i_{k}}^{S} + 1 = k + 1
\end{align*}
By induction, the statement $m_{i_k}^{S} = k$ is true. Thus,
\begin{align*}
m_{i_{n}}^{S} = m_{T}^{S} =  n
\end{align*}
This concludes the proof.
\end{proof}

\begin{lemma}\label{l:v-assign-v-st}  For a pair of hosts $(S,T) \in H \times H$, $v_{i}^{S,T} =1$, if and only if, node $i$ is on the routing path $\pi(S,T) \in P'$.
\end{lemma}
\begin{proof}
Let $\pi(S,T) = \{S, i_{n-1}, i_{n-2}, ..... , i_{1}, T \}$. We first prove that for each node $i \in \pi(S,T)$, $v_{i}^{S,T}$ equals 1. We prove by induction. By definition $v_{T}^{S,T}$ equals 1. Constraints in Equation \ref{eq:node-on-path} ensure, 
\begin{align*}
v_{i_1}^{S,T} = \sum_{j \in V} v_j^{S,T} s_{i_1,j}^S,
\end{align*}
We know $v_T^{S,T}$ equals 1; and $s_{i_1,T}^S$ equals 1 as link $(i_1,T)$ exists on path $\pi(S,T)$.  Consequently, for $j = i_T$, the above equation ensures $v_{i_1}^{S,T}$ equals 1. Now for some $k \in [1, n-1)$ lets assume  $v_{i_k}^{S,T}$ equals 1. Constraints in (\ref{eq:node-on-path}) ensure,
\begin{align*}
v_{i_{k+1}}^{S,T} = \sum_{j \in V} v_{j}^{S,T} s_{i_{k+1},j}^S.
\end{align*}
Now we know $v_{i_{k}}^{S,T}$ equals 1; and $s_{i_{k+1},k}^S$ equals 1 as link $(i_{k+1},k)$ exists on path $\pi(S,T)$. Consequently, for $j = i_k$, the above equation ensures $v_{i_{k+1}}^{S,T}$ equals 1. Thus by induction $v_{k}^{S,T}$ equals 1 for all $k \in [1, n)$.

We now prove that for a node $k \notin \pi(S,T)$, $v_{k}^{S,T}$ equals $0$. We prove by contradiction. Constraints in (\ref{eq:node-on-path}) ensure,
 \begin{align*}
v_{k}^{S,T} = \sum_{j \in V} v_{j}^{S,T} s_{k,j}^S,
\end{align*}
 now if $v_{k}^{S,T}$ equals $1$, there exists some node $j$ such that $v_{j}^{S,T}$ equals 1 (i.e., there exists a link $(i, j)$ which is on the path $\pi(S,T)$); and $s_{k,j}^S$ equals 1 (there exists a link $(i_{k}, j)$ which is not on the path $\pi(S,T)$). Thus there exists two incoming links at node $j$, which violates the constraints in (\ref{eq:src-tree}).  Thus, for any node $k \notin \pi(S,T)$, $v_{k}^{S,T}$ equals 0. This concludes the proof.
\end{proof}

\begin{theorem} \label{th:psm2} 

Given PSMs and DMs as inputs, OCCAM infers a network $N'=(G',H,P')$; such that the routing paths P' satisfy the following properties:
\begin{enumerate}
\item The set of routing paths $P'$ contains an unique acyclic path between each pair of hosts; and
\item $G'$ and $P'$ satisfies all the given PSM and DM constraints.
\end{enumerate} 
\end{theorem}
\begin{proof} 
\begin{enumerate}
\item
Consider a pair of hosts $(S,T) \in H \times H$. The corresponding routing path $\pi(S,T) \in P'$ is constructed in Line 10 of Algorithm GRAPH-CONSTRUCT-I. For each $(S,T) \in H \times H$, the algorithm starts at destination $T$ and builds the path back towards source $S$. A node $i_0$ is initiated to $T$ in line 5 of the algorithm. In each iteration $j \in [1, k]$ of the while loop, a node $i_j$ is found such $s_{i_j, i_{j-1}}^S$ equals 1. Now there exists a path $\pi(i_j,T)$ which is a subpath of path $\pi(S,T)$. We will prove by induction that $\pi(i_j,T)$, exists, is unique and is acyclic. In the first iteration of the while loop, a node $i_1$ is found such that $s_{i_1,T}^S$ equals 1, i.e, link $(i_1, T)$ is on any of the paths from S. Constraints in (\ref{eq:boundary-presence-dst}),
\begin{align*}
\sum_{j \in V} s_{j, T}^S = 1,
\end{align*}
 ensures the presence of such a link. Thus $\pi(i_1,i_0)$ exists, is unique and is acyclic. Now for some $j \in [1, k]$, let's assume  $\pi(i_j,i_0)$ exists, is unique and is acyclic. Constraints in (\ref{eq:boundary-path-constraint}),
 \begin{align*}
s_{i_j,i_{j-1}}^S \leq \sum_{i \in V} s_{i,i_j}^S,
\end{align*}
ensure there exists an incoming link at node $i_j$ from some node $i_{j+1}$ such that $s_{i_{j+1},i_{j}}^S$ equals 1. Now constraints in (\ref{eq:src-tree}) ensures there exists at most 1 incoming link at node $j$; thus path $\pi(i_{j+1}, i_0)$ is unique. The path is acyclic because otherwise Lemma \ref{l:cycle-free} is violated. Thus $\pi(i_{j+1},i_0)$ exists, is unique and is acyclic. 

We now show that there exists some $k$ such that in the $k^{th}$ iteration, $i_k = S$ and the loop terminates. As the number of nodes, i.e, $|V|$ is finite, there must exist some $k$ such that in the $k^{th}$ iteration,
\begin{enumerate}
\item Node $i_k$ is already on the path $\pi(S,T)$. This is not possible as we know $\pi(i_k, i_0)$ is acyclic, or
\item Node $i_k = S$ and the loop terminates.
\end{enumerate}
Thus the path $\pi(S,T)$ constructed in line 10 of  Algorithm GRAPH-CONSTRUCT-I is a path from S to T, and $\pi(S,T)$ is acyclic and unique.
\item Given that Lemma \ref{l:v-assign-mst}, \ref{l:v-assign-v-st} hold, constraints in  (\ref{eq:relative}) and (\ref{eq:entanglement}) ensures P satisfy the given PSM and DM constraints.
\end{enumerate}
\end{proof}

%% file: appendix2.tex
\section{Correctness of Tree Stitching}
\label{app:correctness2}
\begin{lemma}\label{l:segment-path} If there exists a link $(i_1,i_2)$ that belongs to a segment $s \in {\mathcal T}^S$, then there exists a set of links $L^s$ that forms a path from node $i_1$ to some node $i_n \in V$ that marks the end of segment $s$ in $G$. If $s$ terminates at host $T \in {\mathcal T}^S$, then $i_n = T$. 
\end{lemma}
\begin{proof}
If there exists a link $(i_1,i_2)$ belonging to segment $s$ (i.e, $p_{i_1, i_2}^s$ equals 1), then constraint in (\ref{eq:logical-trees}),
\begin{align*}
p_{i_1,i_2}^{s} \leq \sum_{k \in V} p_{i_2, k}^{s} +  \frac{1}{|{\mathcal O}(b_s)|} \sum_{s' \in {\mathcal O}(b_s)} \sum_{k \in V} p_{i_2,k}^{s'} = b_{i_2}^s
\end{align*}
is satisfied if RHS of the inequality is 1. RHS is a sum of two terms ($\sum_{k \in V} p_{i_2, k}^{s}$) and  ($\frac{1}{|{\mathcal O}(b_s)|} \sum_{s' \in {\mathcal O}(b_s)} \sum_{k \in V} p_{i_2,k}^{s'}$ ). Thus,  either (i) there exists an outgoing link $(i_2,i_3)$ for some $i_3 \in V$ that belongs to the same segment $s$ ie. $p_{i_2, i_3}^s = 1$ (satisfying the first term) or, (ii) $i_2$ marks the end of the segment and there exist outgoing links which belong to $s' \in O(s)$ (satisfying the second term). If the first case is true, constraint in (\ref{eq:logical-trees}) imposes the same constraints on link $(i_2, i_3)$. If the second case is true, node $i_2$ marks the end of segment $s$. Thus at each step, we either find an outgoing link $(i_n, i_{n+1})$ that belongs to segment $s$ or node $i_n$ marks the end of segment $s$. At each step, if the outgoing link belongs to the same segment $s$, link $(i_n, i_{n+1})$ would form a cycle at some finite $n$. Now, constraints in (\ref{eq:src-consistency}) ensure that if link $(i_n, i_{n+1})$ belongs to segment $s$, then $s_{i_n, i_{n+1}}^S$ equals $1$. Lemma \ref{l:cycle-free} ensures such links do not form a cycle. Thus, there exists a finite $n$ such that segment $s$ terminates at node $i_n$.

Now, if there exists a link $(i_1,i_2)$ that belongs to a segment $s$ that terminates at some host $T \in H$ (i.e, $p_{i_1, i_2}^s$ equals 1), then constraint in (\ref{eq:logical-trees-end}), 
\begin{align*}
p_{i_1,i_2}^{s} \leq \sum_{k \in V} p_{i_2,k}^{s} = b_{i_2}^s
\end{align*}
ensures there exists an outgoing link $(i_2, i_3)$ that belongs to segment $s$. Now, constraints in (\ref{eq:logical-trees-end}) imposes the same constraint on link $(i_2, i_3)$ and there exists an outgoing link at node $i_3$ that belongs to segment $s$. Thus, at each step if there exists an incoming link $(i_{n-1}, i_{n})$ that belongs to segment $s$, then there exists an outgoing link $(i_{n}, i_{n+1})$ that belongs to segment $s$. Now, as the number of nodes , i.e., $|V|$ is finite, there must exist some $n$ such that,
\begin{enumerate}
\item Node $i_{n}$ already belongs to segment $s$. This is not possible as it forms a cycle and violates Lemma \ref{l:cycle-free}, or
\item Node $i_n = T$.  In this case, constraints in (\ref{eq:logical-trees-end}) does not require the presence of an outgoing link at node $i_n$.
\end{enumerate}
This concludes the proof.
\end{proof} 

\begin{lemma}\label{l:segment-branchpoint} If node $i_n \in V$ marks the end of segment $s$ in $G'$, then $i_n$ marks the start of each segment $s' \in O(s)$, i.e., there exist outgoing links at node $i_n$, such that each outgoing link belongs to a segment $s' \in O(s)$.
\end{lemma}
\begin{proof}
As $i_n$ marks the end of segment $s$, there exists a link $(i_{n-1}, i_n)$ that belongs to segment $s$ (i.e., $p_{i_{n-1}, i_n}^S = 1$). Now, consider the constraint in (\ref{eq:logical-trees}),
\begin{align*}
p_{i_{n-1},i_n}^{s} \leq \sum_{k \in V} p_{i_n, k}^{s} +  \frac{1}{|{\mathcal O}(b_s)|} \sum_{s' \in {\mathcal O}(b_s)} \sum_{k \in V} p_{i_n,k}^{s'} = b_{i_n}^s
\end{align*}
As the number of incoming links at $i_n$ belonging to segment $s$ evaluates to 1, LHS in (\ref{eq:logical-trees}) equals 1. Now the RHS in the above inequality (\ref{eq:logical-trees}) is equated to a binary variable, thus RHS must evaluate to 1. The RHS is a sum of two parts. The first part counts the number of outgoing links at node $i_n$ which belongs to segment $s$. And the second part counts the number of outgoing links at node $i_n$, which belong to any of the segments $s' \in O(s)$, and divides the value by $|O(s)|$. The first part equals 0 as $i_n$ marks the end of segment $s$, forcing the second part to evaluate to $1$. This implies that there must exist $|O(s)|$ number of outgoing links at node $i_n$, and each such link should belong to any of the segments $s' \in O(s)$. Let $L^S$ be a set of such links. Now each link in $L^S$ belongs to a unique segment $s' \in O(s)$, i.e., (i) link $l \in L^S$ belongs to at most one segment (as enforced by constraints (\ref{eq:src-consistency})), (ii) no two links $l_1$ and $l_2$ in $L$ belongs to the same segment $s' \in O(s)$, (as enforced by constraints in (\ref{eq:one-out-link})). Thus there exist $|O(s)|$ outgoing links at node $i_n$, and each such link belongs to a unique segment $s' \in O(s)$.
\end{proof}

\begin{comment}
\begin{lemma}\label{l:last-segment} If there exists a link $(i_1,i_2)$ that belongs to a segment $s \in {\mathcal T}^s$, where $s$ is a segment that terminates at some host $T \in H$, then there exists a set of links $L^s$ that belongs to segment $s$, and links in $L^{s}$ forms a path from node $i_1$ to T.
\begin{proof}
\end{proof}
\end{lemma}
\end{comment}

\begin{theorem} \label{th:tree2} Given source trees and DMs as inputs, OCCAM infers a network $N'=(G',H,P')$; such that the routing paths P' satisfy the following properties
\begin{enumerate}
\item Each routing path $\pi(S,T) \in P'$ is an acyclic path from host $S$ to $T$; and
\item G' and P' is \textit{consistent} with the given source trees and DM measurement inputs.
\end{enumerate} 
\end{theorem}
\begin{proof}
Algorithm GRAPH-CONSTRUCT-II constructs the network $N=(G,H,P)$ using the values assigned by the optimization to variables in Figure \ref{table:variables}. For a host $S \in H$, Lines 3 to 7 in Algorithm GRAPH-CONSTRUCT-II finds a set of links $P^S$ that belong to any of the segments $s \in {\mathcal T}^S$. We will show that links in $P^S$ contains a path from $S$ to each host $T \in H$, and $P^S$ is \textit{consistent} with the source tree ${\mathcal T}^S$.  Let $s_1$ be the segment in ${\mathcal T}^S$ that originates at source host $S$. Now, constraints in (\ref{eq:src-boundary}) ensure there exists a link $(S, k)$ for some $k \in V$ such that link $(S,k)$ belongs to segment $s_1$, i.e. , variable $p_{S,k}^{s_1}$ equals 1. In the presence of such a link, Lemma \ref{l:segment-path} proves there exists a set of links that forms a path $\pi(S, i_{n})$ from host $S$ to some node $i_{n}$, such that each link in $\pi(S, i_{n})$ belongs to segment $s_1$. The path  $\pi(S, i_{n})$ corresponds to segment $s_1$, and branch point $b_{s_1}$ can be mapped to node $i_{n}$. Now Lemma \ref{l:segment-branchpoint} ensures there exist $|O(s_1)|$ outgoing links at node $i_n$ such that that each such link belongs to a unique segment $s' \in O(s_1)$. Lemma \ref{l:segment-path} can now be applied for each segment $s' \in O(s)$ to find a path $\pi(i_{n}, i_{n_s'})$ corresponding to segment $s'$, and branch point $b_{s'}$ can be mapped to node $i_{n_{s'}}$. If segment $s'$ ends at some host $T \in H$,  then Lemma \ref{l:segment-path} proves there exists a set of links that form the path $\pi(i_{n_{s'}}, T)$ which corresponds to segment $s'$. Thus, we showed that each segment $s \in {\mathcal T}^S$ corresponds to a path in $G'$, and each branch point $b_s \in {\mathcal B}^S$ can be mapped to a node $v \in V'$ making $P^S$ consistent with ${\mathcal T}^S$.

Now as $P^S$ forms a tree rooted at host $S$ with each $T \in H$ as a leaf, the path from source $S$ to each host $T \in H$ is unique and acyclic. Lemma \ref{l:v-assign-mst} still holds and thus constraints in Equation \ref{eq:relative} ensure that G' and P' is \textit{consistent} with the DM measurements. 
\end{proof}

%% file: ms.bbl
\begin{thebibliography}{10}

\bibitem{Anandkumar2011}
A.~Anandkumar, A.~Hassidim, and J.~Kelner.
\newblock Topology discovery of sparse random graphs with few participants.
\newblock In {\em Proceedings of the ACM SIGMETRICS Joint International
  Conference on Measurement and Modeling of Computer Systems}, SIGMETRICS '11,
  pages 293--304, New York, NY, USA, 2011. ACM.

\bibitem{Berkolaiko2018}
G.~Berkolaiko, N.~Duffield, M.~Ettehad, and K.~Manousakis.
\newblock Graph reconstruction from path correlation data.
\newblock {\em arXiv preprint arXiv:1804.04574}, 2018.

\bibitem{Bu2002}
T.~Bu, N.~Duffield, F.~L. Presti, and D.~Towsley.
\newblock Network tomography on general topologies.
\newblock In {\em ACM SIGMETRICS Performance Evaluation Review}, volume~30,
  pages 21--30. ACM, 2002.

\bibitem{caceres1999multicast}
R.~C{\'a}ceres, N.~G. Duffield, J.~Horowitz, and D.~F. Towsley.
\newblock Multicast-based inference of network-internal loss characteristics.
\newblock {\em IEEE Transactions on Information theory}, 45(7):2462--2480,
  1999.

\bibitem{CAIDA2015}
CAIDA.
\newblock Archipelago measurement infrastructure.
  "\url{http://www.caida.org/projects/ark/}", 2015.

\bibitem{Castro2004}
R.~Castro, M.~Coates, G.~Liang, R.~Nowak, and B.~Yu.
\newblock Network tomography: Recent developments.
\newblock {\em Statistical science}, pages 499--517, 2004.

\bibitem{Claffy}
K.~Claffy and S.~McCreary.
\newblock Caida skitter project
  \url{https://www.caida.org/tools/measurement/skitter/}.

\bibitem{Coates2003}
M.~Coates, M.~Rabbat, and R.~Nowak.
\newblock Merging logical topologies using end-to-end measurements.
\newblock In {\em Proceedings of the 3rd ACM SIGCOMM conference on Internet
  measurement}, pages 192--203. ACM, 2003.

\bibitem{Donnet2004}
B.~Donnet, P.~Raoult, T.~Friedman, and M.~Crovella.
\newblock Efficient algorithms for large-scale topology discovery.
\newblock {\em CoRR}, cs.NI/0411013, 2004.

\bibitem{Duffield2006}
N.~Duffield, F.~L. Presti, V.~Paxson, and D.~Towsley.
\newblock Network loss tomography using striped unicast probes.
\newblock {\em IEEE/ACM Transactions on Networking}, 14(4):697--710, 2006.

\bibitem{Duffield2002}
N.~G. Duffield, J.~Horowitz, F.~L. Presti, and D.~Towsley.
\newblock Multicast topology inference from measured end-to-end loss.
\newblock {\em IEEE Transactions on Information Theory}, 48(1):26--45, 2002.

\bibitem{Duffield2004}
N.~G. Duffield and F.~L. Presti.
\newblock Network tomography from measured end-to-end delay covariance.
\newblock {\em IEEE/ACM Transactions on Networking (TON)}, 12(6):978--992,
  2004.

\bibitem{eriksson2007learning}
B.~Eriksson, P.~Barford, R.~Nowak, and M.~Crovella.
\newblock Learning network structure from passive measurements.
\newblock In {\em Proceedings of the 7th ACM SIGCOMM conference on
  measurement}, pages 209--214. ACM, 2007.

\bibitem{Gunes2009}
M.~H. Gunes and K.~Sarac.
\newblock Resolving {IP} aliases in building traceroute-based internet maps.
\newblock {\em IEEE/ACM Trans. Netw.}, 17(6):1738--1751, Dec. 2009.

\bibitem{hakimi}
S.~L. Hakimi and S.~S. Yau.
\newblock Distance matrix of a graph and its realizability.
\newblock {\em Quarterly of Applied Mathematics}, 22(4):305--317, 1965.

\bibitem{Knight2011}
S.~Knight, H.~X. Nguyen, N.~Falkner, R.~Bowden, and M.~Roughan.
\newblock The internet topology zoo.
\newblock {\em IEEE Journal on Selected Areas in Communications},
  29(9):1765--1775, 2011.

\bibitem{Mirkovic10the}
J.~Mirkovic, T.~V. Benzel, T.~Faber, R.~Braden, J.~T. Wroclawski, and
  S.~Schwab.
\newblock The deter project: Advancing the science of cyber security
  experimentation and test.
\newblock In {\em In Technologies for Homeland Security (HST), 2010 IEEE
  International Conference on}, page~7, 2010.

\bibitem{Moy}
J.~Moy.
\newblock The ospf protocol, rfc 2328,
  \url{https://tools.ietf.org/html/rfc2328}.
\newblock 1998.

\bibitem{Postel1981}
J.~Postel et~al.
\newblock Rfc 791: Internet protocol, \url{https://tools.ietf.org/html/rfc791},
  1981.

\bibitem{presti2002multicast}
F.~L. Presti, N.~G. Duffield, J.~Horowitz, and D.~Towsley.
\newblock Multicast-based inference of network-internal delay distributions.
\newblock {\em IEEE/ACM Transactions On Networking}, 10(6):761--775, 2002.

\bibitem{Rabbat2004}
M.~Rabbat, R.~Nowak, and M.~Coates.
\newblock Multiple source, multiple destination network tomography.
\newblock In {\em INFOCOM 2004. Twenty-third AnnualJoint Conference of the IEEE
  Computer and Communications Societies}, volume~3, pages 1628--1639. IEEE,
  2004.

\bibitem{Ratnasamy1999}
S.~Ratnasamy and S.~McCanne.
\newblock Inference of multicast routing trees and bottleneck bandwidths using
  end-to-end measurements.
\newblock In {\em INFOCOM'99. Eighteenth Annual Joint Conference of the IEEE
  Computer and Communications Societies. Proceedings. IEEE}, volume~1, pages
  353--360. IEEE, 1999.

\bibitem{Occam}
O.~Razor.
\newblock Principle of occam's razor.
\newblock \url{https://simple.wikipedia.org/wiki/Occam%27s_razor}.

\bibitem{Shavitt2005}
Y.~Shavitt and E.~Shir.
\newblock {DIMES:} let the internet measure itself.
\newblock {\em Computer Communication Review}, 35(5):71--74, 2005.

\bibitem{Spring2004}
N.~Spring, R.~Mahajan, D.~Wetherall, and T.~Anderson.
\newblock Measuring isp topologies with rocketfuel.
\newblock {\em IEEE/ACM Transactions on Networking}, 12(1):2--16, Feb. 2004.

\bibitem{Trassare2013}
S.~T. Trassare, R.~Beverly, and D.~Alderson.
\newblock A technique for network topology deception.
\newblock In J.~Senftle, M.~Beltrani, and K.~Karwedsky, editors, {\em 32th
  {IEEE} Military Communications Conference, {MILCOM} 2013, San Diego, CA, USA,
  November 18-20, 2013}, pages 1795--1800. {IEEE}, 2013.

\bibitem{vardi1996network}
Y.~Vardi.
\newblock Network tomography: Estimating source-destination traffic intensities
  from link data.
\newblock {\em Journal of the American statistical association},
  91(433):365--377, 1996.

\bibitem{Yao2003}
B.~Yao, R.~Viswanathan, F.~Chang, and D.~Waddington.
\newblock Topology inference in the presence of anonymous routers.
\newblock In {\em INFOCOM 2003. Twenty-Second Annual Joint Conference of the
  IEEE Computer and Communications. IEEE Societies}, volume~1, pages 353--363.
  IEEE, 2003.

\end{thebibliography}
